\documentclass[11pt,a4paper,reqno]{amsart}
\usepackage{amsfonts}
\usepackage{amsthm}
\usepackage[T1]{fontenc}
\usepackage{amsmath} 
\usepackage{graphicx}
\usepackage{enumitem}
\usepackage{dsfont}
\usepackage{braket}
\usepackage{amsxtra}
\usepackage{hyperref}
\usepackage[utf8]{inputenc}
\usepackage[capitalize]{cleveref}

\newtheorem{de}{Definition}
\newtheorem{theo}[de]{Theorem}    
\newtheorem{prop}[de]{Proposition}
\newtheorem{lem}[de]{Lemma}
\newtheorem{rem}[de]{Remark}

\newtheorem{exa}[de]{Example}

\DeclareMathOperator{\tr}{Tr}

\DeclareMathOperator{\supp}{supp}

\DeclareMathOperator{\loc}{loc}

\newcommand{\floor}[1]{\lfloor #1 \rfloor}
\newcommand{\myp}[1]{\left(#1\right)}
\newcommand{\normL}[2]{\left\lVert #1 \right\rVert_{L^{#2}(\mathbb{R}^d)}}
\newcommand{\myt}[1]{\left\lbrace #1\right\rbrace}
\newcommand{\abs}[1]{\left\lvert #1 \right\rvert}
\newcommand{\id}{\, \mathrm{d}}

\newcommand{\Esc}{ \mathcal{E}_{\rm Vla}}
\newcommand{\Ecan}{ \mathcal{E}_{\rm Can}}
\newcommand{\ecan}{ e_{\rm Can}^\beta}

\newcommand{\Veff}{ V^{\textrm{eff}}}

\newcommand{\Musc}{\mu_{\textrm{Vla}}}
\newcommand{\esc}{e_{\rm Vla}^\beta}
\newcommand{\per}{\mathrm{per}}
\newcommand{\Ssc}{S_{\textrm{Vla}}}

\newcommand{\ii}{\infty}
\newcommand{\eps}{\varepsilon}
\newcommand\R{{\ensuremath {\mathbb R} }}

\newcommand\nn{\nonumber}
\renewcommand{\epsilon}{\varepsilon}

\newcommand{\norm}[1]{ \left| \! \left| #1 \right| \! \right| }
\newcommand{\gS}{\mathfrak{S}}

\usepackage{xcolor}
\usepackage{fixme}
\fxsetup{
    status=final,
    author=,
    layout=inline,
    theme=color
}
\definecolor{fxnote}{rgb}{0.8000,0.0000,0.0000}
\colorlet{fxnotebg}{yellow}

\title[Fermionic systems at positive temperature]{Semi-classical limit of large fermionic systems at positive temperature}
\date{\today}

\author[M. Lewin]{Mathieu Lewin}
\address{CNRS \& CEREMADE, Paris-Dauphine University, PSL University, 75016 Paris, France} 
\email{mathieu.lewin@math.cnrs.fr}

\author[P.S. Madsen]{Peter S. Madsen}
\address{Department of Mathematics, Aarhus University, Ny Munkegade 118, DK-8000 Aarhus C, Denmark} 
\email{psm@math.au.dk}

\author[A. Triay]{Arnaud Triay}
\address{CEREMADE, Paris-Dauphine University, PSL University, 75016 Paris, France} 
\email{triay@ceremade.dauphine.fr}

\begin{document}
\begin{abstract}
	We study a system of $N$ interacting fermions at positive temperature in a confining potential. In the regime where the intensity of the interaction scales as $1/N$ and with an effective semi-classical parameter $\hbar=N^{-1/d}$ where $d$ is the space dimension, we prove the convergence to the corresponding Thomas-Fermi model at positive temperature. 
\end{abstract}

\maketitle

\setcounter{tocdepth}{1}
\tableofcontents

In this article we study mean-field-type limits for a system of $N$ fermions at temperature $T>0$ in a fixed confining potential. We assume that the interaction has an intensity of the order $1/N$ and that there is an effective semi-classical parameter $\hbar=N^{-1/d}$ where $d$ is the space dimension. In the limit we obtain the nonlinear Thomas-Fermi problem at the corresponding temperature $T>0$. This paper is an extension of a recent work~\cite{FouLewSol-18} by Fournais, Solovej and the first author where the case $T=0$ was solved.

Physically, the Thomas-Fermi model is a rather crude approximation of quantum many-body systems in normal conditions, and it has to be refined in order to obtain a quantitative description of their equilibrium properties. However, certain physical systems in extreme conditions are rather well described by Thomas-Fermi theory. It then becomes important to take into account the effect of the temperature. For instance, the positive-temperature Thomas-Fermi model has been thoroughly studied for very heavy atoms~\cite{FeyMetTel-49,GilPee-55,Latter-55,CowAsh-57,NarThi-81}. It has also played an important role in astrophysics, where the very high pressure encountered in the core of neutron stars and white dwarfs makes it valuable for all kinds of elements of the periodic table~\cite{MarBet-40,March-55b,BruBucJorLom-68,BarBuc-81}. Finally, the Thomas-Fermi model is also useful for ultracold dilute atomic Fermi gases, but the interaction often becomes negligible due to the Pauli principle, except in the presence of spin or of several interacting species~\cite{GoiPitStr-08}. 

In the regime considered in this paper, a mean-field scaling is coupled to a semi-classical limit. This creates some mathematical difficulties. Before~\cite{FouLewSol-18}, this limit has been rigorously considered at $T=0$ for atoms by Lieb and Simon in~\cite{LieSim-77b,LieSim-77} and for pseudo-relativistic stars by Lieb, Thirring and Yau in~\cite{LieThi-84,LieYau-87}. Upper and lower bounds on the next order correction have recently been derived in~\cite{HaiPorRex-18,BenNamPorSchSei-18}, for particles evolving on the torus. For atoms the positive Thomas-Fermi model was derived for the first time in~\cite{NarThi-81}. There are several mathematical works on the time-dependent setting~\cite{NarSew-81,Spohn-81,BarGolGotMau-03,ElgErdSchYau-04,FroKno-11,BenPorSch-14,BenJakPorSafSch-16,BacBrePetPicTza-15,PetPic-16,BenPorSafSch-16,DieRadSch-18}, in which the Schr\"odinger dynamics has been proved to converge to the Vlasov time-dependent equation in the limit $N\to\ii$. Finally, the first two terms in the expansion of the (free) energy of a Fermi gas with spin in the limit $\rho\to0$ was provided in~\cite{LieSeiSol-05} at $T=0$ and in~\cite{Seiringer-06b} at $T>0$. 

The mean-field limit at positive temperature for fermions is completely different from the bosonic case. It was proved in~\cite{LewNamRou-14} that in the similar mean-field regime for bosons, the leading order is the same at $T>0$ as when $T=0$. Only the next (Bogoliubov) correction depends on $T$~\cite{LewNamSerSol-15}. In order to observe an effect of the temperature at the leading order of the bosonic free energy, one should take $T\sim N$, a completely different limit where nonlinear Gibbs measures arise~\cite{Gottlieb-05,LewNamRou-15,LewNamRou-18a,LewNamRou-18b,LewNamRou-18c,Rougerie-15}. Without statistics (boltzons), the temperature does affect the leading order of the energy~\cite{LewNamRou-16}, and the same happens for fermions, as we will demonstrate.

Our method for studying the Fermi gas in the coupled mean-field/semi-classical limit relies on previous techniques introduced in~\cite{FouLewSol-18}. Assuming that the interaction is positive-type ($\widehat{w}\geq0$), the lower bound follows from using coherent states and inequalities on the entropy.  We discuss later in Remark~\ref{rmk:inequality_entropy} a conjectured inequality on the entropy of large fermionic systems which would imply the result for any interaction potential, not necessarily of positive-type. The upper bound is slightly more tedious. The idea is to construct a trial state with locally constant density in small boxes of side length much larger than $\hbar$, and to use the equivalence between the canonical and grand-canonical ensembles for the free Fermi gas. Finally, the convergence of states requires the tools recently introduced in \cite{FouLewSol-18} based on the classical de Finetti theorem for fermions.

The article is organized as follows. In the next section we introduce both the $N$-particle quantum Hamiltonian and the positive-temperature Thomas-Fermi theory which is obtained in the limit. We then state our main theorems, \cref{theo:TL_Can_Sc} and \cref{theo:dilute_limit}. As an intermediate result for the upper bound, we show in Section ~\ref{sec:trial_states} how to approximate a classical density by an $N$ body quantum state. In Section~\ref{sec:proof_non_interacting}, we use this trial state and some known results about the free Fermi gas at positive temperature to prove our main result in the non-interacting case. The interacting case is dealt with in Section~\ref{sec:proof_interacting}. Finally, in Section ~\ref{sec:proof_theo_Vlasov_func} we study the Gibbs state and the minimizers of the Thomas-Fermi functional at positive temperature (\cref{theo:min_Vlasov}).

\subsubsection*{\textbf{Acknowledgment}}
We thank Robert Seiringer for useful comments. 
This project has received funding from the European Research Council (ERC) under the European Union's Horizon 2020 research and innovation programme (grant agreement MDFT No 725528 of M.L.). This work was started when P.S.M. was a visiting student at the University Paris-Dauphine. 
P.S.M. was partially supported by the Sapere Aude grant DFF--4181-00221 from the Independent Research Fund Denmark.

\section{Models and main results}
	\label{sec:non_interacting_case}

\subsection{The Vlasov and Thomas-Fermi functionals at $T>0$}
For a given density $\rho>0$ and an inverse temperature $ \beta > 0 $, the Vlasov functional at positive temperature is given by
\begin{align}
	\mathcal{E}_{\mathrm{Vla}}^{\beta,\rho} \myp{m} 
	= &\frac{1}{\myp{2 \pi}^d} \iint_{\mathbb{R}^{2d}} \myp{|p+A(x)|^2+V\myp{x}} m \myp{x,p} \, \mathrm{d} x \, \mathrm{d} p \nn\\
	&\qquad+\frac{1}{2 \rho} \iint_{\mathbb{R}^{2d}} w \myp{x-y} \rho_m \myp{x} \rho_m \myp{y} \, \mathrm{d} x \, \mathrm{d} y  \nn\\
	 &\qquad+\frac{1}{\myp{2 \pi}^d \beta} \iint_{\mathbb{R}^{2d}} s \myp{m \myp{x,p}} \, \mathrm{d} x \, \mathrm{d} p,\label{eq:semiclfunct}
\end{align}
where $ s \myp{t} = t \log t + \myp{1-t} \log \myp{1-t} $ is the fermionic entropy, and
\begin{equation*}
	\rho_m \myp{x} = \frac{1}{\myp{2 \pi}^d} \int_{\mathbb{R}^d} m \myp{x,p} \, \mathrm{d} p
\end{equation*}
is the spatial density of particles. Here $m$ is a positive measure on the phase space $\R^d\times\R^d$, with the convention
\begin{equation*}
\frac{1}{\myp{2 \pi}^d} \iint_{\R^{2d}} m(x,p)\,\mathrm{d} x\,\mathrm{d} p = \int_{\mathbb{R}^d} \rho_{m}(x) \mathrm{d} x = \rho,
\end{equation*}
and which is assumed to satisfy Pauli's principle $0\leq m\leq 1$. For convenience we have added the factor $1/\rho$ in front of the interaction energy, because it will naturally arise in the mean-field limit. 
We denote the Vlasov minimum free energy by
\begin{equation}
\boxed{e_{\mathrm{Vla}}^{\beta} \myp{\rho} = \inf_{\substack{0\leq m\leq 1\\ (2\pi)^{-d}\iint_{\R^{2d}}m=\rho}} \mathcal{E}_{\mathrm{Vla}}^{\beta,\rho} \myp{m}.}
\label{eq:def_min_Vlasov}
\end{equation}
Precise assumptions on $A,V$ and $w$ will be given later. 

Similarly as in the case $T=0$, we can rewrite the minimum as a two-step procedure where we first choose a density $\nu\in L^1(\R^d,\R_+)$ with $\int_{\R^d}\nu=\rho$ and minimize over all $m$ such that $\rho_m=\nu$, before minimizing over $\nu$. For any fixed constants $\nu\in\R_+$ and $A\in \R^d$  we can solve the problem at fixed $x$ and obtain 
\begin{align*}
\min_{\substack{0\leq m(p)\leq 1\\ (2\pi)^{-d}\int_{\R^{d}}m(p)\,{\rm d}p=\nu}}\left(\frac{1}{\myp{2 \pi}^d} \int_{\mathbb{R}^{d}} |p+A|^2 m \myp{p} \,  \mathrm{d} p +\frac{1}{\myp{2 \pi}^d \beta} \int_{\mathbb{R}^{d}} s \myp{m \myp{p}} \, \mathrm{d} p\right) \\
=-\frac{1}{(2\pi)^d\beta}\int_{\R^d}\log\left(1+e^{-\beta\big(p^2-\mu_{\rm FG}(\beta,\nu)\big)}\right)\,{\rm d}p+\mu_{\rm FG}(\beta,\nu)\,\nu
\end{align*}
where $\mu_{\rm FG}(\beta,\nu)$ is the unique solution to the implicit equation
$$\frac{1}{(2\pi)^d}\int_{\R^d}\frac{1}{1+e^{\beta(p^2-\mu_{\rm FG}(\beta,\nu))}}\,{\rm d}p=\nu$$
and with the unique corresponding minimizer
$$m_{\nu,A}(p)=\frac{1}{1+e^{\beta(|p+A|^2-\mu_{\rm FG}(\beta,\nu))}}.$$
This is the uniform Fermi gas at density $\nu>0$. 
For later purposes we introduce the free energy of the Fermi gas 
\begin{equation}
F_\beta(\nu):=-\frac{1}{(2\pi)^d\beta}\int_{\R^d}\log\left(1+e^{-\beta(p^2-\mu_{\rm FG}(\beta,\nu))}\right)\,{\rm d}p+\mu_{\rm FG}(\beta,\nu)\,\nu.
\label{eq:def_Phi}
\end{equation}
Note that $A$ only appears in the formula of the minimizer. It does not affect the value of the minimum $F_\beta(\nu)$.

All this allows us to reformulate the Vlasov minimization problem using only the density, which leads to the Thomas-Fermi minimization problem at positive temperature $T=1/\beta$
\begin{multline}
e_{\mathrm{Vla}}^{\beta} \myp{\rho}=\min_{\substack{\nu\in L^1(\R^d,\R_+)\\ \int_{\R^d}\nu(x)\,{\rm d}x=\rho}}\bigg\{
\int_{\R^d}F_\beta\big(\nu(x)\big)\,{\rm d}x+\int_{\R^d}V(x)\nu(x)\,{\rm d}x\\+\frac{1}{2 \rho} \iint_{\mathbb{R}^{2d}} w \myp{x-y} \nu \myp{x} \nu \myp{y} \, \mathrm{d} x \, \mathrm{d} y \bigg\}.
\label{eq:def_min_TF}
\end{multline}
The Vlasov minimization~\eqref{eq:def_min_Vlasov} on phase space will be more tractable and we will almost never use the Thomas-Fermi formulation \eqref{eq:def_min_TF} of the problem. 

Now we discuss the existence of a unique Vlasov minimizer for~\eqref{eq:def_min_Vlasov}, under appropriate assumptions on $V,A,w$. We use everywhere the notation $V_\pm=\max(\pm V,0)$ for the positive and negative parts of $V$, which are both positive functions by definition.

\begin{theo}[Minimizers of the Vlasov functional]
	\label{theo:min_Vlasov}
	Fix $ \rho,\beta_0 > 0 $.
	Suppose that $ V_- \in L^{d/2} \myp{\mathbb{R}^d}\cap L^{1+d/2}(\mathbb{R}^d)$, $A\in L^{1}_{\loc}(\mathbb{R}^d)$ and that $V_+ \in L_{\loc}^1 \myp{\mathbb{R}^d} $ satisfies $ \int_{\mathbb{R}^d} e^{-\beta_0 V_+\myp{x}} \, \mathrm{d} x < \infty $. Let
	$$w\in L^{1+\frac{d}{2}} \myp{\mathbb{R}^d} + L_{\varepsilon}^{\infty} \myp{\mathbb{R}^d}+\R_+\delta_0.$$
	Then, for all $\beta > \beta_0$, there are minimizers for the Vlasov problem~\eqref{eq:def_min_Vlasov}. Any minimizer $ m_0$ solves the nonlinear equation
	\begin{equation}
	\label{eq:m0def}
		m_0 \myp{x,p} = \frac{1}{1+ \exp\Big(\beta (|p+A(x)|^2 + V\myp{x} + \rho^{-1} w \ast \rho_{m_0} \myp{x} -\mu) \Big)},
	\end{equation}
for some Lagrange multiplier $\mu$.
	The minimum can be expressed in terms of $m_0$  and $\mu$ as
\begin{align}
		e_{\mathrm{Vla}}^{\beta} \myp{\rho}=
		& -\frac{1}{\myp{2 \pi}^d \beta} \iint_{\mathbb{R}^{2d}} \log \myp{1 + e^{- \beta \big(|p|^2 + V\myp{x} + \rho^{-1} w \ast \rho_{m_0} \myp{x} - \mu\big)}} \, \mathrm{d} x \, \mathrm{d} p \nn \\
		\label{eq:functminimum}
		& + \mu \rho  - \frac{1}{2 \rho} \iint_{\mathbb{R}^{2d}} w \myp{x-y} \rho_{m_0} \myp{x} \rho_{m_0} \myp{y} \, \mathrm{d} x \, \mathrm{d} y.
\end{align}

Furthermore, if $ \widehat{w} \geq 0 $, then $ \mathcal{E}_{\mathrm{Vla}}^{\beta,\rho} $ is strictly convex and therefore has a unique minimizer. In this case, for $\rho'>0$ define
\begin{equation}
F_{\mathrm{Vla}}^{\beta} (\rho,\rho') := \inf_{\substack{0\leq m\leq 1\\ (2\pi)^{-d}\iint_{\R^{2d}}m=\rho}} \mathcal{E}_{\mathrm{Vla}}^{\beta,\rho'} \myp{m}.
\end{equation}
Then, for any $\rho' >0$, $F_{\mathrm{Vla}}^{\beta} (\cdot,\rho')$ is $C^1$ on $\R_+$ and the multiplier appearing in~\eqref{eq:m0def} is given by
\begin{equation}
\label{eq:chemical_potential_derivative}
\mu=\frac{\partial F_{\mathrm{Vla}}^{\beta}}{\partial\rho} (\rho,\rho')_{\big| \rho' = \rho}.
\end{equation}
\end{theo}

The proof of \cref{theo:min_Vlasov} is classical and given for completeness in \cref{sec:proof_theo_Vlasov_func}. Note that the magnetic potential $A$ has only a trivial effect on the minimization problem. The minimizers for a given $A$ are exactly equal to the $m_0(x,p+A)$ with $m_0$ a minimizer for $A\equiv0$. The value of the minimal energy, the density $\rho_{m_0}$ and the Lagrange multiplier $\mu$ are unchanged under this transformation.

The two conditions $e^{-\beta V_+} \in L^1(\mathbb{R}^d)$ and $V_- \in L^{d/2}(\mathbb{R}^d)\cap L^{1+d/2}(\mathbb{R}^d)$ have been chosen to ensure that the minimizer has a finite total mass and a finite total energy. This is because
\begin{align}
	\label{eq:computation_rho}
\iint_{\mathbb{R}^{2d}} \frac{1}{1+e^{\beta(p^2 +V_+ - V_-)}}  
	&\leq \iint_{\mathbb{R}^{2d}} e^{-\beta(p^2/2 + V_+)} + |\{p^2 \leq 2 V_-\}|  \nn \\
	&\leq C\int_{\mathbb{R}^d} \left(\beta^{-\frac{d}2}e^{-\beta V_+} + V_-^{\frac{d}2} \right)
\end{align}
and, similarly,
\begin{align*}
&\iint_{\mathbb{R}^{2d}} \log\left(1+e^{-\beta(p^2 +V_+ - V_-)}\right)\,dx\,dp \\
&\qquad \qquad \leq \iint_{\mathbb{R}^{2d}} e^{-\beta(p^2/2 + V_+)} + C \iint_{\{p^2 \leq 2 V_-\}}\left(1+\beta V_-\right)\\
&\qquad\qquad  \leq C\int_{\mathbb{R}^d} \left(\beta^{-\frac{d}2}e^{-\beta V_+} +V_-^{\frac{d}2} + \beta V_-^{1+\frac{d}2}\right).
\end{align*}

\subsection{The $N$-body Gibbs state and its limit}
The aim of this paper is to understand the large--$N$ limit of fermionic systems in a mean-field-type regime. We will end up with the Vlasov problem~\cref{eq:semiclfunct} introduced in the previous section. 

\subsubsection{The mean-field limit}

Here we analyze the `mean-field' limit where the interaction has a fixed range and a small intensity. We consider the following Hamiltonian
\begin{equation}
\label{eq:hamiltonian}
	H_{N,\hbar} = \sum\limits_{j=1}^N  |i\hbar \nabla_{x_j} + A(x_j)|^2 + V\myp{x_j} + \frac{1}{N} \sum\limits_{1\leq j < k\leq N} w \myp{x_j - x_k}
\end{equation}
acting on the Hilbert space $ \bigwedge^N_1 L^2 \myp{\mathbb{R}^d} $ of anti-symmetric functions. For simplicity we neglect the spin variable. 
We suppose that 
$$|A|^2, w \in L^{1+\frac{d}{2}} \myp{\mathbb{R}^d} + L_{\varepsilon}^{\infty} \myp{\mathbb{R}^d}$$
and that $w$ is an even function. We also assume that the electric potential $ V \in L_{\textrm{loc}}^{1+{d}/{2}} \myp{\mathbb{R}^d} $ is confining, that is, $ V \myp{x} \to \infty $ when $ \left\lvert x \right\rvert \to \infty $, and that the divergence is so fast that $ \int e^{-\beta_0 V_+\myp{x}} \, \mathrm{d} x < \infty $ for some $\beta_0>0$. Note that this implies that $V_-$ has a compact support, hence in particular $V_-\in L^{d/2}(\R^d)\cap L^{1+d/2}(\mathbb{R}^d)$.
At inverse temperature $ \beta > \beta_0 $, the canonical free energy is given by the functional
\begin{equation}
	\mathcal{E}_{\mathrm{Can}}^{N,\hbar} \myp{\Gamma} = \tr \myp{H_{N,\hbar} \Gamma} + \frac{1}{\beta} \tr (\Gamma \log \Gamma),
\end{equation}
defined for all fermionic quantum states $\Gamma=\Gamma^*\geq0$ with $\tr(\Gamma)=1$. The minimum over all $\Gamma$ is uniquely attained at the Gibbs state 
\begin{equation*}
\Gamma_{N,\hbar,\beta} = Z^{-1} e^{-\beta H_{N,\hbar}},
\end{equation*}
where $ Z = \tr e^{-\beta H_{N,\hbar}} $, which leads to the minimum free energy
\begin{align*}
\boxed{e_{\mathrm{Can}}^{\beta} \myp{\hbar,N}
	:= \min_{\Gamma} \mathcal{E}_{\mathrm{Can}}^{N,\hbar} \myp{\Gamma} 
	= - \frac{1}{\beta} \log \tr e^{-\beta H_{N,\hbar}}.}
\end{align*}

Our main result is the following.

\begin{theo}[Mean-field limit]
	\label{theo:TL_Can_Sc}
Let $ \beta_0, \rho > 0 $. Assume that $ V \in L_{\loc}^{1+d/2} \myp{\mathbb{R}^d} $ is such that $ V \myp{x} \to \infty $ at infinity and that $ \int e^{-\beta_0 V_+\myp{x}} \, \mathrm{d} x < \infty $.
Furthermore, assume $ |A|^2,w \in L^{1+d/2} \myp{\mathbb{R}^d} + L_{\varepsilon}^{\infty} \myp{\mathbb{R}^d} $ with $w$ even and satisfying 
$\widehat{w} \geq 0$. Then, for all $\beta > \beta_0$ we have the convergence
\begin{equation}
\boxed{\lim_{\substack{N\to\infty \\ \hbar^d N \to \rho}} \hbar^d \ecan(\hbar,N) = \esc(\rho).}
\label{eq:limit}
\end{equation}
Moreover, if $(\Gamma_N)$ is a sequence of approximate Gibbs states so that
\begin{equation}
 \Ecan^{N,\hbar}(\Gamma_N) = \ecan(\hbar,N) + o(1), 
\label{eq:approx_min}
 \end{equation}
then the one particle density of $\Gamma_N$ satisfies the following convergence
\begin{equation*}
\hbar^d \rho_{\Gamma_N}^{\myp{1}} \rightharpoonup \rho_{m_0} \textrm{ weakly in } L^1(\mathbb{R}^d) \cap L^{1+2/d}(\mathbb{R}^d),
\end{equation*}
and
\begin{align}
& m_{f,\Gamma_N}^{(1)} \longrightarrow m_0 \textrm{ strongly in } L^1(\mathbb{R}^{2d}), \label{eq:limit_1PDM_a}\\
& \rho_{m_{f,\Gamma_N}^{(1)}} \longrightarrow \rho_{m_0} \textrm{ strongly in } L^1(\mathbb{R}^d) \cap L^{1+2/d}(\mathbb{R}^d),\label{eq:limit_1PDM_b}
\end{align}
where $m_{f,\Gamma_N}^{(k)}$ is the $k$-particle Husimi function of $\Gamma_N$ and $m_0$ is the unique minimizer of the Vlasov functional in \cref{eq:m0def}.
The $k$-particle Husimi functions converge weakly in the sense that 
\begin{equation}
\label{eq:cv_of_states_husimi}
\int_{\mathbb{R}^{2dk}} m_{f,\Gamma_N}^{(k)} \varphi  \to \int_{\mathbb{R}^{2dk}} m^{\otimes k}_0 \varphi 
\end{equation}
for all $\varphi \in L^{1}(\mathbb{R}^{2dk}) + L^{\infty}(\mathbb{R}^{2dk})$. Similarly, if we denote by $\mathcal{W}_{\Gamma_{N}}^{(k)}$ the $k$-particle Wigner measure of $\Gamma_{N}$, we also have,
\begin{equation}
\label{eq:cv_of_states_wiener}
\int_{\mathbb{R}^{2dk}} \mathcal{W}_{\Gamma_{N}}^{(k)} \varphi \to \int_{\mathbb{R}^{2dk}} m^{\otimes k}_0 \varphi,
\end{equation}
for all $\varphi$ satisfying $\partial_{x_1}^{\alpha_1}...\partial_{x_k}^{\alpha_k}\partial_{p_1}^{\beta_1}...\partial_{p_k}^{\beta_k} \varphi \in L^\infty(\mathbb{R}^{2dk})$, where $\max(\alpha_j,\beta_j) \leq 1$.

\end{theo}

The Husimi function $m_{f,\Gamma_N}^{(k)}$ (based on a given shape function $f$) and the Wigner measure $\mathcal{W}_{\Gamma_{N}}^{(k)}$ are defined and studied at length in~\cite{FouLewSol-18}. These are some natural semiclassical measures that can be associated with $\Gamma_N$ in the $k$-particle phase space $\R^{2dk}$. We will recall their definition in the proof later in Section~\ref{sec:CV_states}. 

\begin{rem}
For simplicity we work with a confining potential $V$ but Theorems~\ref{theo:min_Vlasov} and~\ref{theo:TL_Can_Sc} hold the same when $\mathbb{R}^d$ is replaced by a bounded domain $\Omega$ with any boundary conditions.
\end{rem}

\begin{rem}
Our lower bound relies on the strong assumption that $\widehat{w} \geq 0$, but the upper bound does not. It is classical that a positive Fourier transform allows to easily bound the interaction from below by a one-body potential, see  \cref{ineq:w_hat_pos} below. 
\end{rem}

\begin{rem}
As mentioned in~\eqref{eq:approx_min}, we can handle approximate Gibbs states with a free energy close to the minimum with an error $o(1)$, although the energy is itself of order $N$. Our proof actually applies to the one-particle Husimi function under the weaker condition that $\Ecan^{N,\hbar}(\Gamma_N) = \ecan(\hbar,N) + o(N)$ but our argument does not easily generalize to higher order Husimi functions. Of course, for the exact quantum Gibbs state we have equality $\Ecan^{N,\hbar}(\Gamma_{N,\hbar,\beta}) = \ecan(\hbar,N)$.
\end{rem}

\begin{rem}
\label{rmk:inequality_entropy}
Without the assumption $\widehat{w} \geq 0$, the Vlasov functional $\Esc^{\beta,\rho}$ can have several minimizers and the limit in \cref{eq:cv_of_states_husimi} is believed to be an average over the set of minimizers of $\Esc^{\beta,\rho}$. Namely there exists a so called de Finetti measure $\mathcal{P}$~\cite{FouLewSol-18}, concentrated on the set of minimizers for $e^\beta_{\rm Vla}$, such that
\begin{equation*}
m_{f,\Gamma_N}^{(k)}  \rightharpoonup \int m^{\otimes k} {\rm d}\mathcal{P}(m),
\end{equation*}
in the sense defined in Theorem~\ref{theo:TL_Can_Sc}. We conjecture the following Fatou-type inequality on the entropy
\begin{equation}
\liminf_{\substack{N\to \infty \\ \hbar^d N \to \rho}} \hbar^d \tr \Gamma_N \log \Gamma_N \geq \int \left(\int_{\mathbb{R}^{2d}} s(m) \right)\,{\rm d}\mathcal{P}(m)
\label{eq:Fatou_entropy}
\end{equation}
for general sequences $(\Gamma_N)$ with de Finetti measure $\mathcal{P}$. Should this inequality be true, we could remove the assumption $\widehat{w} \geq 0$ in \cref{theo:TL_Can_Sc}. In fact, in our proof we show that the above inequality holds when the right-hand side is replaced by $$\int_{\mathbb{R}^{2d}} s\left( \int m \,{\rm d}\mathcal{P}(m)\right).$$ 
When there is a unique minimizer, the two coincide.
\end{rem}

\begin{exa} [Large atoms in a strong harmonic potential]\rm
The Hamiltonian in \cref{eq:hamiltonian} can describe a large atom in a strong harmonic potential. Indeed, consider $N$ electrons in a harmonic trap and interacting with a nucleus of charge $Z$. In the Born-Oppenheimer approximation, the $N$ electrons are described by the Hamiltonian
\begin{equation*}
\sum_{j=1}^N -\Delta_{x_j} + \omega^2 |x_j|^2 - \frac{Z}{|x_j|} + \sum_{j<k} \frac{1}{|x_j-x_k|}.
\end{equation*}
Scaling length in the manner $x_j = N^{-1/2}x_j'$ we see that this Hamiltonian is unitarily equivalent to
\begin{equation*}
N^{4/3} \left(\sum_{j=1}^N -N^{-\frac{2}{3}}\Delta_{x_j} + \left(\omega N^{-1}\right)^2 |x_j|^2 - \frac{Z N^{-1}}{|x_j|} + \frac{1}{N} \sum_{j<k} \frac{1}{|x_j-x_k|} \right).
\end{equation*}
Hence taking $Z$ proportional to $N$ and $\omega$ proportional to $N$, we obtain the Hamiltonian of \cref{eq:hamiltonian} with $d=3$, $A = 0$, $V(x) = |x|^2$ and $w(x) = |x|^{-1}$. In the limit we find the positive-temperature Thomas-Fermi model for an atom in a harmonic trap, which has stimulated many works in the Physics literature~\cite{FeyMetTel-49,GilPee-55,Latter-55,CowAsh-57}. This convergence has been proved for the first time by Narnhofer and Thirring in~\cite{NarThi-81}, but starting from the grand-canonical model instead of the canonical ensemble as we do here. 
\end{exa}

\subsubsection{The dilute limit}

In this section we deal with the case where the interaction potential has a range depending on $N$ and tending to zero in our limit $N\to\ii$ with $\hbar^d N\to\rho$. This is classically taken into account by choosing the interaction in the form
\begin{equation}
w_N (x) := N^{d\eta} w(N^\eta x)
\end{equation}
for a fixed $w$ and a fixed parameter $\eta>0$. In our confined system, the average distance between the particles is of order $N^{-1/d}\simeq \rho^{-1/d}\hbar$. The system is dilute when the particles interact rarely, that is, $\eta>1/d$. For bosons in 3D, the limit involves the finite-range interaction $4\pi a \delta_0$ where $a=\int_{\mathbb{R}^d}w / (4\pi)$ for $\eta<1$ and $a=a_s$, the $s$-wave scattering length $a_s$ when $\eta=1$. Due to the anti-symmetry the $s$-wave scattering length does not appear for fermions, except if there are several different species, e.g.~with spin. This regime has been studied in \cite{LieSeiSol-05} for the ground state and \cite{Seiringer-06b} at positive temperature, for the infinite translation-invariant gas. Here we extend these results to the confined case but do not consider any spin for shortness, hence we obtain a trivial limit.  
Our main result for dilute systems is the following.

\begin{theo}[Dilute limit]
\label{theo:dilute_limit}
Let $ \beta_0, \rho > 0 $. We assume that $ V \in L_{\loc}^{1+d/2} \myp{\mathbb{R}^d} $ is such that $ V \myp{x} \to \infty $ at infinity and that $ \int e^{-\beta_0 V_+\myp{x}} \, \mathrm{d} x < \infty $. Furthermore, assume that $ |A|^2 \in L^{1+d/2} \myp{\mathbb{R}^d} + L_{\varepsilon}^{\infty} \myp{\mathbb{R}^d} $ and $w\in L^1(\mathbb{R}^d) \cap L^{1+d/2}(\mathbb{R}^d)$ is even.

\medskip

\noindent$\bullet$ If $0< \eta < 1/d$ and $\widehat{w}\geq 0$ then, for all $\beta > \beta_0$ we have
\begin{equation*}
\lim_{\substack{N\to\infty \\ \hbar^d N \to \rho}} \hbar^d \ecan(\hbar,N) = \textrm{e}^{\beta,(\int_{\mathbb{R}^d} w)\delta_0}_{\mathrm{Vla}}(\rho)
\end{equation*}
where $\textrm{e}^{\beta,(\int_{\mathbb{R}^d} w)\delta_0}_{\mathrm{Can}}(\rho)$ is the minimum of the Vlasov energy with interaction potential $(\int_{\mathbb{R}^d} w)\delta_0$.  

\medskip

\noindent$\bullet$ If $\eta > 1/d$, $d\geq3$ and $w \geq 0$ is compactly supported, then for all $\beta > \beta_0$ we have
\begin{equation*}
\lim_{\substack{N\to\infty \\ \hbar^d N \to \rho}} \hbar^d \ecan(\hbar,N) = \textrm{e}^{\beta,0}_{\mathrm{Vla}}(\rho)
\end{equation*}
where $\textrm{e}^{\beta,0}_{\mathrm{Can}}(\rho)$ is the minimum of the Vlasov energy without interaction potential.

\medskip

In both cases, we have the same convergence of approximate Gibbs states as in Theorem~\ref{theo:TL_Can_Sc}.
\end{theo}

The proof of Theorem~\ref{theo:dilute_limit} is given in Section \ref{sec:proof_interacting}.

\section{Contruction of trial states}
\label{sec:trial_states}

In this section we construct a trial state for the proof of the upper bound. In the dilute case this construction is similar to the one in \cite{Seiringer-06b} where the thermodynamic limit of non-zero spin interacting fermions were studied in the grand-canonical picture. In particular we will make use of \cite[Lemma 2]{Seiringer-06b}. Precisely we prove the following proposition.

\begin{prop}[Canonical trial states]
\label{prop:trial_states}
Let $\rho_0 \in C^\infty_c(\mathbb{R}^d)$ be such that $\int_{\mathbb{R}^d} \rho_0 = 1$. Assume $|A|^2 \in L^{1+d/2}(\mathbb{R}^d)$, $w \in L^1(\mathbb{R}^d)\cap L^{1+d/2}(\mathbb{R}^d)$. If  $\eta d > 1$, we assume $w$ to be compactly supported. Then, there is a sequence of canonical states $\Gamma_N$ on $\bigwedge_{i=1}^N L^2(\mathbb{R}^d)$ satisfying
\begin{equation}\label{eq:trial_state_cv_energy}
\hbar^d \left(\tr \sum_{i=1}^N (i\hbar \nabla + A)^2 \Gamma_N + \tr \Gamma_N \log \Gamma_N \right)\underset{\substack{N\to\infty\\\hbar^d N \to 1}}{\longrightarrow} \int_{\mathbb{R}^d} F_\beta(\rho_0),
\end{equation}
\begin{equation}\label{eq:trial_state_cv_onebody}
\left\| \frac{1}{N} \rho_{\Gamma_N}^{(1)} - \rho_0\right\|_{L^1(\mathbb{R}^d)} \underset{\substack{N\to\infty\\\hbar^d N \to 1}}{\longrightarrow} 0
\end{equation} 
and
\begin{equation}\label{eq:trial_state_cv_twobody}
\frac{\hbar^d}{N}\int_{\mathbb{R}^{2d}} w_N(x-y) \rho_{\Gamma_N}^{(2)}(x,y) dx dy \underset{\substack{N\to\infty\\\hbar^d N \to 1}}{\longrightarrow} 
	\left\{
		\begin{array}{ll}
		\int_{\mathbb{R}^d} (w \ast \rho_0) \rho_0  & \mathrm{ if } \; \eta = 0 \\[3pt]
		\left(\int_{\mathbb{R}^d} w\right) \int_{\mathbb{R}^d} \rho_0^2 & \mathrm{ if } \;  0 < d \eta < 1 \\[3pt]
		0 & \mathrm{ if } \; d\eta > 1, d\geq 3.
		\end{array}
		\right.
\end{equation}
Furthermore, we can take $\rho_{\Gamma_N}^{(1)}$ to be supported in a compact set which is independent of $N$ and uniformly bounded in $L^{\infty}(\mathbb{R}^d)$ so that the convergence (\ref{eq:trial_state_cv_onebody}) holds in fact in all $L^p(\mathbb{R}^d)$ for $1 \leq p < \infty$.
\end{prop}

\begin{proof}
The proof consists in dividing the space into small cubes in which we take a correlated version of the minimizer for the free case and then do the thermodynamic limit in these cubes. This choice allows us to control the one-body density, which will be almost constant in these boxes. Without loss of generality, we will write the proof for $A=0$. The proof is the same for $A\neq 0$.

\subsection*{Step 1. Definition of the trial state}

Let $\rho_0 \in C^{\infty}_c(\mathbb{R}^d)$ and take $R>0$ such that $\supp \rho_0 \subset [-R/2,R/2)^d =: C_R$. Divide $C_R$ in small cubes of size $\ell >0$, $C_R \subset \bigcup_{z\in B_\infty(R\ell^{-1}) \cap \mathbb{Z}^d}  \Lambda_z$ with $\Lambda_z := z\ell + [-\ell/2,\ell/2)^d$. We will take later $1\gg \ell \gg \hbar$. For all $z$ define $N_z := \floor{\hbar^d \ell^d \min_{\Lambda_z}\rho_0}$ so that $\sum_z N_z \leq N$. For $0< \varepsilon < \ell/4$ and for all $z$, define the box 
$$\widetilde{\Lambda}_z := z \ell + \left[-\frac{\ell-\varepsilon}2,\frac{\ell-\varepsilon}2\right)^d \subset \Lambda_z $$ 
and denote by 
\begin{equation*}
\widetilde{\Gamma}_z = \frac{e^{-\beta\left(\sum_{i=1}^{N_z} -\hbar^2 \Delta^{\per}_i \right)}}{Z_z}
	= \sum_{k\in \mathcal{P}_{N_z}(\mathbb{Z}^d)} \lambda_k \ket{e_{k_1}\wedge \cdots \wedge e_{k_{N_z}}}\bra{e_{k_1}\wedge \cdots \wedge e_{k_{N_z}}}
\end{equation*}
the canonical minimizer of the free energy at inverse temperature $\beta$ of $N_z$ free fermions in the box $\widetilde{\Lambda}_z$ with periodic boundary conditions, where $\mathcal{P}_{n}(E)$ denotes the set of all subset of $E$ with $n$ elements. For $j\in \mathbb{Z}^d$, $$e_j(x) = (\ell-\varepsilon)^{-d/2} e^{i \frac{2\pi }{\ell-\varepsilon} j \cdot x}$$ are the eigenfunctions of the periodic Laplacian in $\widetilde{\Lambda}_z$ and $\lambda_k$ the eigenvalues of $\widetilde{\Gamma}_z$ associated with $e_k= e_{k_1}\wedge ... \wedge e_{k_{N_z}}$. Note that we omit the $z$ dependence of $\lambda_k$ and $e_k$. We now regularize these functions and construct a state in the slightly larger cube $\Lambda_z$ with Dirichlet boundary condition. Let $\chi \in C^\infty(\mathbb{R}^d)$ such that $\chi \equiv 0$ in $\mathbb{R}^d \setminus B(0,1)$, $\chi \geq 0$ and $\int_{\mathbb{R}^d} \chi = 1$, denote $\chi_\varepsilon = \varepsilon^{-d} \chi(\varepsilon^{-1}\cdot)$ and define for $j\in \mathbb{Z}^d$
\begin{equation*}
f_j := e_j \sqrt{\mathds{1}_{\widetilde{\Lambda}_z} \ast \chi_\varepsilon}.
\end{equation*}
Note that
\begin{align*}
\int_{\Lambda_z} f_j \overline{f_k} 
&= \int e_j \overline{e_k} \, (\mathds{1}_{\widetilde{\Lambda}_z} \ast \chi_\varepsilon) = \int \int_{\widetilde{\Lambda}_z}  e_j(x) \overline{e_k}(x) \chi_\varepsilon(y-x) \mathrm{d} y \mathrm{d} x \\
&=  \int_{\widetilde{\Lambda}_z}  e_j \overline{e_k} \int_{\mathbb{R}^d} \chi = \delta_{j,k}.
\end{align*}
Hence the family $(f_j)_j$ is still orthonormal and one can check that it satisfies $f_j \equiv e_j$ in $ [-(\ell-2\varepsilon)/2,(\ell-2\varepsilon)/2)^d$ and as well as the Dirichlet boundary condition on $\Lambda_z$. Besides from having a state satisfying the Dirichlet boundary condition, we also want to add correlations in order to deal with the $d\eta > 1$ case. Let $\varphi \in C^\infty_c(\mathbb{R}^d)$ such that $\varphi \equiv 0$ in $B(0,1)$, $\varphi \equiv 1$ in $B(0,2)^c$ and $\varphi \leq 1$ almost everywhere and for $s>0$ denote $\varphi_s = \varphi(s^{-1}\cdot)$. Following \cite{Seiringer-06b}, we define the correlation function $F(x_1,...,x_{N_z}) = \prod_{i<j} \varphi_s(x_i-x_j)$ and the state
\begin{equation*}
\Gamma_z = \sum_{k\in \mathcal{P}_{N_z}(\mathbb{Z}^d)} \lambda_k Z_k^{-1}  \ket{F f_{k_1}\wedge ... \wedge f_{k_{N_z}}}\bra{F f_{k_1}\wedge ... \wedge f_{k_{N_z}}},
\end{equation*}
where $Z_k = \|F f_{k_1}\wedge ... \wedge f_{k_{N_z}} \|_{L^2(\Lambda_z^{N_z})}^2$ are normalization factors. Now consider the state 
$$\Gamma := \bigwedge_{z} \Gamma_z.$$
We will show that $\Gamma$ satisfies the three limits \cref{eq:trial_state_cv_energy},~\eqref{eq:trial_state_cv_onebody} and~\eqref{eq:trial_state_cv_twobody}. This state does not have the exact number of particle $N$ but satisfies $\sum_z N_z = N - \mathcal{O}(\ell N)$. Hence we will only have to correct the particle number by adding $\mathcal{O}(\ell N)$ uncorrelated particles of low energy, for instance outside the support of $\rho_0$. This will not modify the validity of the three limits. Now we focus on $\Gamma$ and compute its free energy. 

In the case $\eta d < 1$, we choose the following regime for the parameters introduced above.
\begin{equation*}
s \ll \hbar \ll \varepsilon \ll \ell \ll N^{-\eta} \textrm{ and } s \ell \ll \hbar^2.
\end{equation*}
One could in fact take $\Gamma_{F=1}$ (removing the factor $F$, see below) and remove the dependence in $s$. In the case $\eta d > 1$, the convergence holds in the regime
\begin{equation*}
N^{-\eta} \ll s \ll \hbar \ll \varepsilon \ll \ell \textrm{ and } s \ell \ll \hbar^2.
\end{equation*}

\subsection*{Step 2. Verification of (\ref{eq:trial_state_cv_energy})}

We fix $z$ and work in the cube $\Lambda_z$. Let us first compute the kinetic energy of the correlated Slater determinants appearing in the definition of $\Gamma_z$ (note that this is not a eigenfunction decomposition due of the lack of orthogonality). Let us denote $X = (\sqrt{\mathds{1}_{\widetilde{\Lambda}_z} \ast \chi_\varepsilon})^{\otimes N_z}$ so that $\Psi_k^z := f_{k_1}\wedge ... \wedge f_{k_{N_z}} = X e_{k_1}\wedge ... \wedge e_{k_{N_z}}$ (we will omit the superscript $z$ when there is no ambiguity) and denote $\nabla, -\Delta$ the gradient and the Laplacian for all coordinates $x_1,...,x_{N_z}$ in the box $\Lambda_z$ with Dirichlet boundary condition, we can check that $$\nabla \left(F X e_{k_1}\wedge ... \wedge e_{k_{N_z}} \right) = \left(X \nabla F + F \nabla X + i FX \sum_{j=1}^{N_z} \frac{2\pi k_j}{\ell-\varepsilon} \right)e_{k_1}\wedge ... \wedge e_{k_{N_z}}.$$
Hence,
\begin{equation*}
\tr (-\Delta) \Gamma_z = \sum_{k\in \mathcal{P}_{N_z}(\mathbb{Z}^d)} \frac{\lambda_k}{Z_k} \left(\epsilon_k+ \|(X \nabla F + F \nabla X)e_{k_1}\wedge ... \wedge e_{k_{N_z}}\|_{L^2(\Lambda_z^{N_z})}^2\right)
\end{equation*}
where $$\varepsilon_k := \left|2\pi (\ell-\varepsilon)^{-1} \sum_{j=1}^{N_z} k_j \right|^2$$ is the eigenvalue of $-\Delta^{\per}$ associated with the eigenfunction $e_k$. Note that $\lambda_k \propto e^{-\beta \hbar^2 \epsilon_k}$. We will show that $\sum_{k} \lambda_k Z_k^{-1} \epsilon_k \simeq \sum_{k} \lambda_k\epsilon_k = \tr (-\Delta^{\per}) \widetilde{\Gamma}$ and that  the second summand above is an error term. For that we first need to estimate the normalization factors $Z_k$ and then bound the factor with the $\nabla F$ and $\nabla X$. We will use several times that for any sequence $a_1,...,a_p >0$ we have
\begin{equation}
	\label{ineq:prod_sum}
1 \geq \prod_{n=1}^p(1-a_n) \geq 1 -  \sum_{n=1}^p a_n.
\end{equation}
Hence,
\begin{align}
	\label{eq:estimate_Zk}
Z_k &= \int_{\Lambda_z^{ N_z}} \prod_{1\leq n < m \leq N_z} \varphi_s(x_n-x_m)^2 |\Psi_k|^2 dX \nn \\
	&\geq 1 - \int_{\Lambda_z^{N_z}} \sum_{1\leq n < m \leq N_z}(1-\varphi_s(x_n-x_m)^2)  |\Psi_k|^2 dX \nn  \\
	&\geq 1 - \int_{\Lambda_z^{2}}(1-\varphi_s(x_1-x_2)^2) \rho^{(1)}_{\Psi_k}(x_1)\rho^{(1)}_{\Psi_k}(x_2) dx_1dx_2 \nn \\
	&\geq 1 - C s^d \ell^d \hbar^{-2d},
\end{align}
where we used that $\rho^{(2)}_{\Psi_k}(x,y) \leq \rho^{(1)}_{\Psi_k}(x)\rho^{(1)}_{\Psi_k}(y)$ because $\Psi_k$ is a Slater determinant, and that $\rho^{(1)}_{\Psi_k} = N_z \ell^{-d} \mathds{1}_{\widetilde{\Lambda}_z} \ast \chi_\varepsilon \leq C \hbar^{-d}$.

Then we compute 
\begin{equation*}
|\nabla_{x_1} F|^2 = \sum_{\substack{m\neq n \\ m,n\geq 2}}^{N_z} \frac{\nabla \varphi_s(x_1 - x_m)\cdot \nabla \varphi_s(x_1 - x_n)}{\varphi_s(x_1- x_m) \varphi_s(x_1- x_m)} F^2 + \sum_{m\geq 2}^{N_z} \frac{|\nabla \varphi_s(x_1 - x_m)|^2}{\varphi_s(x_1- x_m)^2} F^2
\end{equation*}
and obtain 
\begin{align*}
\| \nabla F \Psi_k \|_{L^2(\Lambda_z^{N_z})}^2 &\leq C \int_{\Lambda_z^{3d}} |\nabla \varphi_s(x_1 - x_2)|| \nabla \varphi_s(x_1 - x_3)| \times \\
&\quad\quad\quad\quad\quad\quad \times \rho^{(1)}_{\Psi_k}(x_1)\rho^{(1)}_{\Psi_k}(x_2) \rho^{(1)}_{\Psi_k}(x_3) dx_1 dx_2 dx_3 \\
& \quad + C \int_{\Lambda_z^{2d}} |\nabla \varphi_s(x_1 - x_2)|^2  \rho^{(1)}_{\Psi_k}(x_1)\rho^{(1)}_{\Psi_k}(x_2) dx_1 dx_2 \\
&\leq C s^{-2}\left(s^{2d} \ell^d \hbar^{-3d} + s^d \ell^d \hbar^{-2d} \right).
\end{align*}
Now we turn to the $\nabla X$ part. We have $$\nabla_{x_1} X(x_1,...,x_{N_z}) = \frac{\nabla\sqrt{\mathds{1}_{\widetilde{\Lambda}_z} \ast \chi_\varepsilon}(x_1)}{\sqrt{\mathds{1}_{\widetilde{\Lambda}_z} \ast \chi_\varepsilon(x_1)}} X(x_1,...,x_{N_z})$$ and
\begin{align*}
\int_{\Lambda_z^{N_z}} |\nabla X|^2 |e_{k_1}\wedge ... \wedge e_{k_{N_z}}|^2 
	&= \int_{\Lambda_z^{N_z}}\sum_{j=1}^{N_z}  \left|\frac{\nabla\sqrt{\mathds{1}_{\widetilde{\Lambda}_z} \ast \chi_\varepsilon}(x_1)}{\sqrt{\mathds{1}_{\widetilde{\Lambda}_z} \ast \chi_\varepsilon(x_1)}}\right|^2  |\Psi_k|^2 \\
	&= \int_{\Lambda_z} \left|\frac{\nabla\sqrt{\mathds{1}_{\widetilde{\Lambda}_z} \ast \chi_\varepsilon}(x_1)}{\sqrt{\mathds{1}_{\widetilde{\Lambda}_z} \ast \chi_\varepsilon(x_1)}}\right|^2 \rho^{(1)}_{\Psi_k} \\
	&\leq C  \int_{\Lambda_z} \left| \nabla\sqrt{\mathds{1}_{\widetilde{\Lambda}_z} \ast \chi_\varepsilon}(x_1)\right|^2 N_z \ell^{-d}  \\
	& \leq C N_z \ell^{-d} \int_{\Lambda_z} \int |\nabla \sqrt{\chi_\varepsilon}|^2 \leq C \ell^d \hbar^{-d} \varepsilon^{-2},
\end{align*}
where we used the pointwise bound $ \left| \nabla\sqrt{\mathds{1}_{\widetilde{\Lambda}_z} \ast \chi_\varepsilon}(x_1)\right|^2 \leq  \int |\nabla \sqrt{\chi_\varepsilon}|^2 $.
Since $X$ and $F$ are both bounded by $1$ we obtain
\begin{multline*}
\tr (-\Delta) \Gamma_z = \tr (-\Delta^{\per}) \widetilde{\Gamma}_z + \mathcal{O} \bigg(\frac{ s^{2(d-1)} \ell^d \hbar^{-3d} + s^{d-2} \ell^d \hbar^{-2d} +\ell^d \hbar^{-d} \varepsilon^{-2}}{1 - C s^d \ell^d \hbar^{-2d}} \\ + N_z^{1+2/d} s^d \ell^d \hbar^{-2d} \bigg).
\end{multline*}

We proceed with estimating the entropy of $\Gamma_z$. Thanks to \cite[Lemma 2]{Seiringer-06b} we have
\begin{align*}
\tr \Gamma_z \log \Gamma_z 
	&\leq \tr \widetilde{\Gamma}_z \log \widetilde{\Gamma}_z - \log \min_{k} Z_k \\
	&=\tr \widetilde{\Gamma}_z \log \widetilde{\Gamma}_z  +  \mathcal{O}\left( s^d \ell^d \hbar^{-2d}\right),
\end{align*}
where we used the estimate (\ref{eq:estimate_Zk}) on $Z_k$. Combining the last two estimates gives
\begin{align*}
&\tr (-\hbar^2 \Delta) \Gamma + \tr \Gamma \log \Gamma \\
	&\quad = \sum_{z}\tr (-\hbar^2 \Delta) \Gamma_z + \tr \Gamma \log \Gamma_z \\
	&\quad \leq \sum_{z} e_{\mathrm{Can}}^{\beta,\per}(\widetilde{\Lambda}_{z}, \hbar, N_z) + \ell^{-d}  \mathcal{O}\left( s^d \ell^d \hbar^{-2d}\right)
	+ \mathcal{O}\bigg( (\hbar^{-d}\ell^d)^{1+2/d} s^d \ell^d \hbar^{-2d} \bigg) \\
	&\quad \qquad \qquad \qquad + \hbar^2 \ell^{-d}  \mathcal{O} \left(\frac{ s^{2(d-1)} \ell^d \hbar^{-3d} + s^{d-2} \ell^d \hbar^{-2d} +\ell^d \hbar^{-d} \varepsilon^{-2}}{1 - C s^d \ell^d \hbar^{-2d}}\right),
\end{align*}
where we used that $N_z \leq \|\rho_0\|_{L^\infty(\mathbb{R}^d)} \hbar^{-d}\ell^d$. It is a known fact \cite{Robinson-71,Ruelle} (see also \cite{MadThese,TriThese} for more details) that 
\begin{equation}
e_{\mathrm{Can}}^{\beta,\per}(\widetilde{\Lambda}_{z}, \hbar, N_z) = \hbar^{-d}\ell^d F_\beta(N_z/(\hbar^{-d}\ell^d)) + o(\hbar^{-d}\ell^d)
\end{equation}
locally uniformly in $\rho_z := N_z \hbar^d \ell^{-d}$ as $\hbar \to 0$ under the condition $\hbar \ll \ell$. This is the thermodynamic limit of the free Fermi gas. By the continuity of $F_\beta$ and the estimate $N_z / (\hbar^{-d}(\ell-\varepsilon)^d) = \rho(z) + \mathcal{O}(\varepsilon \ell^{-1})$ we obtain
\begin{multline*}
\hbar^d \big(\tr(- \hbar^2 \Delta) \Gamma + \tr \Gamma \log \Gamma\big) 
	\leq \ell^d \sum_{z\in \mathbb{Z}^d} F_\beta(\rho(z)) + o(1) + \mathcal{O}(\varepsilon/ \ell) \\
	 +\mathcal{O}\left( (s\ell / \hbar^2)^{d} \ell^2 \right) + \mathcal{O}\left(s/\hbar\right)^d +  \mathcal{O} \left(\frac{ (s/\hbar)^{2(d-1)} + (s/\hbar)^{d-2} + (\hbar/\varepsilon)^{2}}{1 - C (s \ell /\hbar^2)^d}\right).
\end{multline*}
If $s \ll \hbar \ll \varepsilon \ll \ell $ with the extra condition that $s \ell \ll \hbar^2$ we obtain the upper bound in (\ref{eq:trial_state_cv_energy}) by passing to the limit and by identifying the first term above as a Riemann sum. The lower bound is obtained in the same fashion by seeing $\Gamma_z$ as a trial state for the periodic case.

\subsection*{Step 3. Verification of (\ref{eq:trial_state_cv_onebody})}
Let us recall that $\Gamma_{F=1}$ is the uncorrelated version of the trial state (which corresponds to taking $\varphi \equiv 1$) and that we denote by $\rho_{F=1}^{(k)}$ its $k$-particle density, for  $k\geq 1$. From (\ref{ineq:prod_sum}) and using that $\Gamma_{F=1}$ is a sum of Slater determinants we have
\begin{align*}
N^{-1}&\|\rho_{\Gamma}^{(1)}-\rho_{F=1}^{(1)}\|_{L^1(\mathbb{R}^d)} \\
	&\leq N^{-1} \sum_{z\in \mathbb{Z}^d} \sum_{k\in \mathcal{P}_{N_z}(\mathbb{Z}^d)}  \lambda_k^z \iint_{\mathbb{R}^{2d}} (1-\varphi_s(x_1-x_2)^2)\rho^{(2)}_{\Psi_k^z}(x_1,x_2) dx_1 dx_2 \\
	&\leq C N^{-1} \sum_{z\in \mathbb{Z}^d} \iint_{\Lambda_z^2} (1-\varphi_s(x_1-x_2)^2) N_z^2 \ell^{-2d} dx_1 dx_2 \\
	&\leq C N^{-1} \sum_{z\in \mathbb{Z}^d} s^d \ell^d N_z^2 \ell^{-2d} \\
	&\leq C (s/\hbar)^d.
\end{align*}
We also used that $\rho^{(2)}_{\Psi_k^z}\leq \rho^{(1)}_{\Psi_k^z}\otimes \rho^{(1)}_{\Psi_k^z} \leq \|\rho_0\|_{L^\infty(\mathbb{R}^d)} N_z^2 \ell^{-2d}$. Finally, denoting by $\Gamma_{z,F=1}$ the uncorrelated version of $\Gamma_z$ and by $\rho_{z,F=1}^{(1)}$ its one-body density we have
\begin{align*}
\|N^{-1} \rho_{F=1}^{(1)} - \rho_0\|_{L^1(\mathbb{R}^d)} 
	&\leq \sum_{z} \|N^{-1}\rho_{z,F=1}^{(1)} - \rho_0\mathds{1}_{\Lambda_z}\|_{L^1(\mathbb{R}^d)} \\
	&\leq C \sum_{z} \left(\|\nabla\rho_0\|_{L^\infty(\mathbb{R}^d)} \ell^{d+1} + \|\rho_0\|_{L^\infty(\mathbb{R}^d)} \ell^{d-1} \varepsilon\right) \\
	&\leq C (\ell + \varepsilon/\ell).
\end{align*}
We have used that in $z\ell + [-(\ell-2\varepsilon)/2, (\ell - 2 \varepsilon)/2)^d$, 
\begin{equation*}
	N^{-1}\rho_{z,F=1}^{(1)} = N^{-1}\ell^{-d} \floor{\hbar^{-d} \ell^d \min_{\Lambda_z}\rho_0} = \rho_0 + \mathcal{O}(\hbar^d \ell^d) + \mathcal{O}(\|\nabla\rho_0\|_{L^\infty(\mathbb{R}^d)} \ell)
\end{equation*}
and that 
\begin{equation*}
	\| N^{-1}\rho_{z,F=1}^{(1)} - \rho_0\mathds{1}_{\Lambda_z}\|_{L^\infty(\Lambda_z)} \leq C \|\rho_0\|_{L^\infty(\mathbb{R}^d)}.
\end{equation*}
Under the stated conditions on $\hbar,\ell,s$ and $\varepsilon$ we have $N^{-1} \rho_{\Gamma}^{(1)} \to \rho_0$ in $L^1(\mathbb{R}^d)$.

\subsection*{Step 4. Verification of (\ref{eq:trial_state_cv_twobody})}
Let us first turn to the case $0 \leq \eta d < 1$. Note that the two-particle density matrices satisfy
\begin{align}
\Gamma^{(2)} &= \sum_{z\in \mathbb{Z}^d} \Gamma_z^{(2)}+ \sum_{z \neq z'}  \Gamma_z^{(1)}\otimes \Gamma_{z'}^{(1)}  \nn \\
	&=  \Gamma^{(1)} \otimes \Gamma^{(1)} + \sum_{z\in \mathbb{Z}^d} \Gamma_z^{(2)} -  \Gamma_z^{(1)}\otimes \Gamma_{z}^{(1)}.
	\label{eq:2_pdm_Gamma}
\end{align}
In particular we obtain for the two-particle reduced density
\begin{equation}
\rho_{\Gamma}^{(2)} =  \rho_{\Gamma}^{(1)} \otimes \rho_{\Gamma}^{(1)} + \sum_{z\in \mathbb{Z}^d} \rho_{\Gamma_z}^{(2)} -  \rho_{\Gamma_z}^{(1)}\otimes \rho_{\Gamma_{z}}^{(1)}.
	\label{eq:2_pdm}
\end{equation}
The second term above is negligible in our regime. Indeed, using the triangle inequality, the Lieb-Thirring inequality~\cite{LieThi-75,LieThi-76} and Young's inequality we obtain
\begin{align*}
N^{-2} \bigg| \sum_{z\in \mathbb{Z}^d} &\int_{\mathbb{R}^d} w_N(x-y) \left(\rho_{\Gamma_z}^{(2)} -  \rho_{\Gamma_z}^{(1)}\otimes \rho_{\Gamma_{z}}^{(1)}\right) \bigg| \\
	&\leq CN^{-2} \|w_N\|_{L^{1+d/2}(\mathbb{R}^d)} \sum_{z\in \mathbb{Z}^d} \Bigg\{ \|\rho_{\Gamma_z}^{(1)}\|_{L^{1+2/d}(\mathbb{R}^d)}\|\rho_{\Gamma_z}^{(1)}\|_{L^{1}(\mathbb{R}^d)} \\
	&\qquad \qquad \qquad \qquad \qquad \qquad + N_z^2 \left(\frac{\tr \Gamma_z \left(\sum_{j=1}^{N_z} -\Delta_{z_j}\right)}{ N_z^{1+2/d}}\right)^{\frac{1}{1+2/d}}  \Bigg\} \\
	&\leq C N^{-2} N^{d\eta \frac{d}{d+2}} \sum_{z\in \mathbb{Z}^d}  N_z^2 \ell^{-\frac{2}{1+2/d}}  \\
	&\leq C (N^{\eta} \ell)^d \ell^{d(1-\frac{1}{d+2})}
\end{align*}
where we used that $\rho_{\Gamma_z}^{(1)} \leq C N_z \ell^{-d} \leq C \|\rho_0\|_{L^\infty(\mathbb{R}^d)} \hbar^{-d}$ almost everywhere and the estimate on the kinetic energy of $\Gamma_z$ computed before. Hence, if $N^{-1}\rho_{\Gamma}^{(1)} \to \rho_0$ in $L^1(\mathbb{R}^d)$ and if $\ell = o(N^{-\eta})$,  since both $N^{-1}\rho_{\Gamma}^{(1)}$ and $\rho_0$ are bounded (uniformly in $N$) in $L^\infty(\mathbb{R}^d)$, by (\ref{eq:2_pdm}) and the use of Young's inequality we obtain (\ref{eq:trial_state_cv_twobody}) for $0\leq \eta d < 1$.

The case $\eta d > 1$ is easier to handle since in this case $N^{-\eta} = o (s)$. Indeed, due to the correlation factor $F$ and because $w$ is compactly supported we will have $\tr w_N(x-y) \Gamma = 0$ for $N$ sufficiently large. 
\end{proof}

\section{Proof of Theorem~\ref{theo:TL_Can_Sc} in the non-interacting case $w\equiv0$}\label{sec:proof_main_thm}
\label{sec:proof_non_interacting}

In this section we prove the convergence~\eqref{eq:limit} of the free energy in \cref{theo:TL_Can_Sc} in the case where the interaction is dropped, that is $w\equiv0$. We study the interacting case later in Section~\ref{sec:proof_interacting}. The convergence of states will be discussed in Section~\ref{sec:CV_states}. 

The non-interacting case is well understood since the Hamiltonian is quadratic in creation and annihilation operators in the grand canonical picture. The minimizers are known to be the so-called quasi-free states~\cite{BacLieSol-94}. For those we have an explicit formula and the argument of the proof is reduced to a usual semi-classical limit. The upper bound on the free energy is a consequence of Proposition~\ref{prop:trial_states} from the previous section. The proof of the lower bound relies on localization and the use of coherent states. 

We start with the following well-known lemma, the proof of which can for instance rely on Klein's inequality and the convexity of the fermionic entropy~$s$ \cite{Thirring}.
\begin{lem}[The minimal free energy of quasi-free states]\label{lem:value_energy_QF-states}
	\label{lem:unrestricted_oneb}
	Let $ \beta > 0 $, and let $ H $ be a self-adjoint operator on a Hilbert space $ \mathfrak{H} $ such that $\tr e^{- \beta H} < \infty $. Then 
	\begin{equation*}
	\min_{\substack{0 \leq \gamma \leq 1 \\ \gamma \in \mathfrak{S}_1 \myp{\mathfrak{H}}} } \myp{\tr H \gamma + \frac{1}{\beta} \tr s \myp{\gamma}} 
	= - \frac{1}{\beta} \tr \log \myp{1+e^{-\beta H}},
	\end{equation*}
	with the unique minimizer being $ \gamma_0 = \frac{1}{1+ e^{\beta H}} $.
\end{lem}

With Lemma~\ref{lem:value_energy_QF-states} at hand we are able to provide the 

\begin{proof}[Proof of \cref{theo:TL_Can_Sc} in the non-interacting case.]
	Suppose that $ w =0 $.
	We start out by proving the upper bound on the energy, using the trial states constructed in the previous section.
	Let $ \rho > 0 $ and $ 0 \leq \nu \in C_c^\infty \myp{\mathbb{R}^d} $ with $ \int_{\mathbb{R}^d} \nu \myp{x} \, \mathrm{d} x = \rho $.
	By \cref{prop:trial_states} we then have a sequence $ \myp{\Gamma_N} $ of canonical $ N $-particle states satisfying
	\begin{equation*}
	\hbar^d \tr \myp{ \sum_{j=1}^N |i\hbar \nabla_{x_j} + A(x_j)|^2 \Gamma_N } + \frac{\hbar^d}{\beta} \tr \Gamma_N \log \Gamma_N
	\to \int_{\mathbb{R}^d} F_{\beta} \myp{\nu \myp{x}} \, \mathrm{d} x.
	\end{equation*}
	The one-particle densities $ \hbar^d \rho_{\Gamma_N}^{\myp{1}} $ converge to $ \nu $ strongly in $ L^{1} \myp{\mathbb{R}^d} $ and are uniformly bounded in $L^\infty(\mathbb{R}^d)$. Hence they converge strongly in all $L^p(\mathbb{R}^d)$ for $p\in [1,\infty)$.
	Since $ V \in L_{\loc}^{1+d/2} \myp{\mathbb{R}^d} $ and $ \rho_{\Gamma_N}^{\myp{1}} $ are, by construction, supported in a fixed compact set, we have
	\begin{equation*}
	\hbar^d \tr V \myp{x} \Gamma_N^{\myp{1}} = \hbar^d \int_{\mathbb{R}^d} V\myp{x} \rho_{\Gamma_N}^{\myp{1}} \myp{x} \, \mathrm{d} x
	\to \int_{\mathbb{R}^d} V\myp{x} \nu \myp{x} \, \mathrm{d} x.
	\end{equation*}
	This means that
	\begin{equation*}
	\hbar^d \ecan \myp{\hbar,N} \leq \hbar^d \mathcal{E}_{\mathrm{Can}}^{N,\hbar} \myp{\Gamma_N} 
	\to \int_{\mathbb{R}^d} F_{\beta} \myp{\nu \myp{x}} \, \mathrm{d} x + \int_{\mathbb{R}^d} V\myp{x} \nu \myp{x} \, \mathrm{d} x,
	\end{equation*}
	and, since $ \nu $ is arbitrary, we have shown that 
	\begin{equation*}
	\limsup_{\substack{N \to \infty \\ \hbar^d N \to \rho}} \hbar^d \ecan \myp{\hbar,N} \leq e_{\mathrm{Vla}}^{\beta} \myp{\rho}.
	\end{equation*}
	
	To prove the lower bound, we use the following bound \cite{BacLieSol-94,Thirring} on the entropy
	\begin{equation*}
		\tr \Gamma \log \Gamma 
		\geq \tr\myp{ \Gamma^{\myp{1}} \log \Gamma^{\myp{1}} + \myp{1-\Gamma^{\myp{1}}} \log \myp{1-\Gamma^{\myp{1}}} }
		= \tr s \myp{\Gamma^{\myp{1}}}
	\end{equation*}
	which follows from the fact that quasi-free states maximize the entropy at given one-particle density matrix $\Gamma^{(1)}$. The bound applies to any $N$-particle state $\Gamma$ whose one-particle density is $\Gamma^{(1)}$.
	Applying \cref{lem:unrestricted_oneb} above, we have for any $ \mu \in \mathbb{R} $ and any $ N $-body state $ \Gamma $
	\begin{align*}
		\Ecan^{N,\hbar} \myp{\Gamma} 
		&\geq \tr \myp{|i\hbar \nabla + A(x)|^2 + V \myp{x} - \mu} \Gamma^{\myp{1}} + \frac{1}{\beta} \tr s \myp{\Gamma^{\myp{1}}} + \mu N \\
		&\geq - \frac{1}{\beta} \tr \log \myp{1 + e^{-\beta \myp{|i\hbar \nabla + A(x)|^2 + V \myp{x} - \mu}}} + \mu N.
	\end{align*}
	Thus, we are left to using the known semi-classical convergence (whose proof is recalled below in \cref{prop:gcsemilimit})
	\begin{multline}
		\liminf_{\hbar \to 0} - \frac{\hbar^d}{\beta} \tr \log \myp{1 + e^{-\beta \myp{|i\hbar \nabla + A(x)|^2 + V \myp{x} - \mu}}} \\
		\geq -\frac{1}{\myp{2 \pi}^d \beta} \iint_{\mathbb{R}^{2d}} \log \myp{1+ e^{-\beta\myp{p^2+V \myp{x} - \mu}}} \, \mathrm{d} x \, \mathrm{d} p,
		\label{eq:known_sc_limit}
	\end{multline}
	and to take $ \mu = \mu_{\mathrm{Vla}} \myp{\rho} $. Recognizing the expression of the Vlasov free energy on the right-hand side we appeal to \cref{theo:min_Vlasov} and immediately obtain
	\begin{equation*}
		\liminf_{\substack{N \to \infty \\ \hbar^d N \to \rho}} \hbar^d \ecan \myp{\hbar,N} \geq e_{\mathrm{Vla}}^{\beta} \myp{\rho},
	\end{equation*}
	concluding the proof of \eqref{eq:limit} in the non-interacting case.
\end{proof}

In~\eqref{eq:known_sc_limit} we have used the following well-known fact, which we prove for completeness. 

\begin{prop}[Semi-classical limit]
	\label{prop:gcsemilimit}
	Let $ \beta_0> 0 $, we assume that $|A|^2\in L^{1+d/2}(\mathbb{R}^d) + L^\infty_\varepsilon(\mathbb{R}^d)$, $ V \in L_{\loc}^{1+d/2} \myp{\mathbb{R}^d} $ is such that $ V \myp{x} \to \infty $ at infinity and that $ \int e^{-\beta_0 V_+\myp{x}} \, \mathrm{d} x < \infty $.	
	Then for any chemical potential $ \mu \in \mathbb{R} $ and all $ \beta > \beta_0 $,
	\begin{multline}
		\limsup_{\hbar \to 0} \frac{\hbar^d}{\beta} \tr \log \myp{1+ e^{-\beta \myp{(|i\hbar \nabla + A|^2 + V-\mu}}} \\
		\leq \frac{1}{\myp{2 \pi}^d \beta} \iint_{\mathbb{R}^{2d}} \log \myp{1+ e^{-\beta \myp{p^2 + V\myp{x} - \mu}}} \, \mathrm{d} x \, \mathrm{d} p.
	\label{eq:semi-classics}
	\end{multline}
\end{prop}

This result is known \cite{Thirring} and the proof we provide here is essentially the one in \cite{Simon-80}, where however the von Neumann entropy $x\log(x)$ was used instead of the Fermi-Dirac entropy $x\log(x)+(1-x)\log(1-x)$. In fact, \cref{theo:TL_Can_Sc} shows that the inequality~\eqref{eq:semi-classics} is indeed an equality.

\begin{proof}[Proof of \cref{prop:gcsemilimit}]
	Without loss of generality we may assume that $ \mu = 0 $. We also assume in a first step that $V_- \in L^{\infty}(\mathbb{R}^d)$ and then remove this assumption at the end of the proof.
	Due to technical issues involving the potential $V$, we need to localize the minimization problem on some bounded set.
	Let $ \chi, \eta \in C^{\infty} \myp{\mathbb{R}^d} $ satisfy $ \chi^2 + \eta^2 = 1 $, $ \supp \chi \subseteq B \myp{0,1} $ and $ \supp \eta \subseteq B \myp{0, \tfrac{1}{2}}^c $.
	For $R>0$, denote $ \chi_R = \chi \myp{\frac{\cdot}{R}} $ and $ \eta_R = \eta \myp{\frac{\cdot}{R}} $.
	Let $ H_{\hbar} = |i\hbar \nabla + A|^2 + V $ and take $ \gamma^{\hbar} = \frac{1}{1+ e^{\beta H_{\hbar}}} $ as in \cref{lem:unrestricted_oneb}.
	By the IMS localization formula we have
	\begin{equation}
	\label{eq:ims}
	\tr H_{\hbar} \gamma^{\hbar} = \tr \myp{H_{\hbar} \chi_R \gamma^{\hbar} \chi_R} + \tr \myp{H_{\hbar} \eta_R \gamma^{\hbar} \eta_R} - \hbar^2 \tr \myp{ \left\lvert \nabla \chi_R \right\rvert^2 + \left\lvert \nabla \eta_R \right\rvert^2}\gamma^{\hbar},
	\end{equation}
	and using the convexity of $s$ and \cite[Theorem 14]{BroKos-90},
	\begin{align}
	\tr s \myp{\gamma^{\hbar}} &= \tr \chi_R s \myp{\gamma^{\hbar}} \chi_R + \tr \eta_R s \myp{\gamma^{\hbar}} \eta_R\nn\\ 
	&\geq \tr s \myp{\chi_R \gamma^{\hbar} \chi_R} + \tr s \myp{\eta_R \gamma^{\hbar} \eta_R}.\label{eq:entropybound}
	\end{align}
	We first deal with the localization outside the ball. 
	The operators we consider in $ B \myp{0, \frac{R}{2} }^c $ are the ones with Dirichlet boundary condition. We obtain by \cref{lem:unrestricted_oneb} that the remainder terms are bounded by
	\begin{align}
	&\tr \myp{H_{\hbar} \eta_R \gamma^{\hbar} \eta_R} + \frac{1}{\beta} \tr s \myp{\eta_R \gamma^{\hbar} \eta_R} \nn \\
	&\qquad\qquad \geq - \frac{1}{\beta} \tr_{L^2( B \myp{0, \frac{R}{2} }^c)} \log \myp{1+e^{-\beta \myp{|i\hbar \nabla + A|^2 + V - C}}} \nonumber \\
	&\qquad\qquad \geq - \frac{C}{\beta} \tr_{L^2( B \myp{0, \frac{R}{2} }^c)} e^{-\beta \myp{|i\hbar \nabla + A|^2  + V}} \nn \\
	&\qquad\qquad \geq - \frac{C}{\beta} \tr_{L^2( B \myp{0, \frac{R}{2} }^c)} e^{-\beta \myp{-\hbar^2\Delta^D  + V}}\label{eq:minmax_2} \\
	& \qquad\qquad \geq - \frac{C}{\beta} \tr_{L^2(\mathbb{R}^d)} e^{-\beta \myp{-\hbar^2 \Delta + (1-\alpha)V + \alpha \inf_{B(0,R)^c}V}}  \label{eq:minmax_3} \\
	&\qquad\qquad\geq -\frac{C e^{-\alpha \inf_{B(0,R)^c}V}}{\myp{2 \pi \hbar}^d} \iint_{\mathbb{R}^{2d}} e^{-\beta \myp{p^2 + (1-\alpha) V \myp{x}}} \, \mathrm{d} x \, \mathrm{d} p,
	\label{eq:gcnonlocal}
	\end{align}
	where $\alpha> 0$ is such that $\beta(1-\alpha)> \beta_0$. The inequality (\ref{eq:minmax_2}) comes from the diamagnetic inequality \cite{CycFroKirSim-87} and (\ref{eq:minmax_3}) is obtained by the min-max characterization of the eigenvalues. The last inequality follows from Golden-Thompson's formula \cite[Theorem VIII.30]{ReeSim1}. 
	
	The error term in the IMS formula can be estimated by 
	\begin{align}
	- \tr \myp{ \left\lvert \nabla \chi_R \right\rvert^2 + \left\lvert \nabla \eta_R \right\rvert^2}\gamma^{\hbar}\nn
	&\geq - \frac{C}{R} \tr \gamma^{\hbar}\nn\\
	&\geq - \frac{C}{R} \tr e^{-\beta H_{\hbar}} \nonumber \\
	&\geq - \frac{C}{R \myp{2 \pi \hbar}^d} \iint_{\mathbb{R}^{2d}} e^{-\beta \myp{p^2 + V \myp{x}}} \, \mathrm{d} x \, \mathrm{d} p,
	\label{eq:gcerror}
	\end{align}
	where we used again the diamagnetic and Golden-Thompson inequalities.

	Next we derive a bound on the densities $ \rho_{\gamma_R^{\hbar}} $, where $ \gamma_R^{\hbar} = \chi_R \gamma^{\hbar} \chi_R $, using the Lieb-Thirring inequality~\cite{LieThi-75,LieThi-76}.
	Combining \labelcref{eq:ims,eq:entropybound,eq:gcerror,eq:gcnonlocal} we have shown
	\begin{multline}
		\tr H_{\hbar} \gamma_R^{\hbar} + \frac{1}{\beta} \tr s \myp{\gamma_R^{\hbar}} - \frac{\varepsilon\myp{R}}{\hbar^d}
		\leq \tr H_{\hbar} \gamma^{\hbar} + \frac{1}{\beta} \tr s\myp{\gamma^{\hbar}} \\
		= -\frac{1}{\beta} \tr \log \myp{ 1 + e^{-\beta H_{\hbar}}} 
		\leq 0
	\label{eq:locerror}
	\end{multline}
	where $ \varepsilon \myp{R} \to 0 $ when $ R \to \infty $.
	By \cref{lem:unrestricted_oneb} we have
	\begin{multline*}
	\tr H_{\hbar} \gamma_R^{\hbar} + \frac{1}{\beta} \tr s \myp{\gamma_R^{\hbar}}
	\geq \frac{1}{2} \tr \myp{-\hbar^2 \Delta} \gamma_R^{\hbar} \\ - \frac{1}{\beta} \tr \log \myp{1 + e^{-\beta \myp{|i\hbar \nabla + A|^2/2 + V}}} - \frac{C}{\hbar^d}
	\end{multline*}
	where, as in \eqref{eq:gcnonlocal},
	\begin{align*}
	\tr \log \myp{1 + e^{-\beta \myp{|i\hbar \nabla + A|^2/2 + V}}}
	\leq \frac{C e^{-\alpha \inf V}}{\myp{2 \pi \hbar}^d} \iint_{\mathbb{R}^{2d}} e^{-\beta\myp{p^2/2 + (1-\alpha) V \myp{x}}} \, \mathrm{d} x \, \mathrm{d} p.
	\end{align*}
	This implies the following bound on the kinetic energy
	\begin{equation}
		\label{eq:bound_kinetic_energy}
	\tr \myp{- \hbar^2 \Delta} \gamma_R^{\hbar} \leq \frac{C}{\hbar^d}.
	\end{equation}
	By the Lieb-Thirring inequality, we obtain
	\begin{equation}
	\label{eq:densitybound}
	\int_{\mathbb{R}^d} \rho_{\gamma_R^{\hbar}} \myp{x}^{1+\frac{2}{d}} \, \mathrm{d} x
	\leq C \tr \myp{- \Delta \gamma_R^{\hbar}}
	\leq \frac{1}{\hbar^{d+2}} C.
	\end{equation}
	
	We return to the estimate on the localized terms in \eqref{eq:ims} and \eqref{eq:entropybound}, using coherent states.
	Let $ f \in C_c^{\infty} \myp{\mathbb{R}^d} $ be a real-valued and even function, and consider the coherent state $ f_{x,p}^{\hbar} \myp{y} = \hbar^{-\frac{d}{4}} f \big( \hbar^{-\frac{1}{2}} \myp{y-x} \big) e^{i \frac{p\cdot y}{\hbar}} $. The projections $ \ket{f_{x,p}^{\hbar}} \bra{f_{x,p}^{\hbar}} $ give rise to a resolution of the identity on $ L^2 \myp{\mathbb{R}^d} $:
	\begin{equation}
	\frac{1}{(2\pi \hbar)^d} \int_{\mathbb{R}^{2d}}\ket{f_{x,p}^{\hbar}} \bra{f_{x,p}^{\hbar}} = \mathrm{Id}_{L^2(\mathbb{R}^d)}.
	\label{eq:coherent_state_repres}
	\end{equation}
	Using this in combination with Jensen's inequality and the spectral theorem, we obtain
	\begin{align}
		\tr s \myp{\chi_R \gamma^{\hbar} \chi_R}
		&= \frac{1}{\myp{2 \pi \hbar}^d} \iint_{\mathbb{R}^{2d}} \left\langle f_{x,p}^{\hbar} , s \myp{\gamma_R^{\hbar}} f_{x,p}^{\hbar} \right\rangle \, \mathrm{d} x \, \mathrm{d} p \nonumber \\
		&\geq \frac{1}{\myp{2 \pi \hbar}^d} \iint_{\mathbb{R}^{2d}} s \myp{ \left\langle f_{x,p}^{\hbar}, \gamma_R^{\hbar} f_{x,p}^{\hbar} \right\rangle } \, \mathrm{d} x \, \mathrm{d} p.
	\label{eq:gclocenergy}
	\end{align}
	On the other hand, applying \cite[Corollary 2.5]{FouLewSol-18} we have
	\begin{align}
		\tr H_{\hbar} &\chi_R \gamma^{\hbar} \chi_R \nn \\
		&= \frac{1}{\myp{2 \pi \hbar}^d} \iint_{\mathbb{R}^{2d}} \left\langle f_{x,p}^{\hbar}, H_{\hbar} \gamma_R^{\hbar} f_{x,p}^{\hbar} \right\rangle \, \mathrm{d} x \, \mathrm{d} p \nonumber \\
		&= \frac{1}{\myp{2 \pi \hbar}^d} \iint_{\mathbb{R}^{2d}} \myp{|p+A|^2 + V \myp{x}} \left\langle f_{x,p}^{\hbar}, \gamma_R^{\hbar} f_{x,p}^{\hbar} \right\rangle \, \mathrm{d} x \, \mathrm{d} p \nonumber \\
		& +  \int_{\mathbb{R}^{d}} \rho_{\gamma_R^{\hbar}} \myp{ A^2 - A^2 \ast \left\lvert f^{\hbar} \right\rvert^2 } \nn  - 2 \Re \tr \myp{ A - A \ast \left\lvert f^{\hbar} \right\rvert^2 } \cdot i\hbar \nabla \gamma_R ^\hbar \nn \\
		& \qquad \qquad \qquad \quad - \hbar \int_{\mathbb{R}^d} \left\lvert \nabla f \right\rvert^2 + \int_{\mathbb{R}^{d}} \rho_{\gamma_R^{\hbar}} \myp{ V - V \ast \left\lvert f^{\hbar} \right\rvert^2 } 
	\label{eq:gclocentropy}
	\end{align}
	Since $ \hbar^d \rho_{\gamma_R^{\hbar}} $ is supported in $ B \myp{0,R} $ and is uniformly bounded in $ L^{1+2/d} \myp{\mathbb{R}^d} $ by \eqref{eq:densitybound}, and $ V \ast \left\lvert f^{\hbar} \right\rvert^2 $ converges to $ V $ locally in $ L^{1+ d/2} \myp{\mathbb{R}^d} $. The same argument applied to $A$ and $|A|^2$ combined with H\"older's inequality, the Lieb-Thirring inequality and (\ref{eq:bound_kinetic_energy}) shows that the remainder terms above are $ o \myp{\hbar^{-d}} $.
	At last, combining \labelcref{eq:locerror,eq:gclocenergy,eq:gclocentropy} as well as a simple adaptation of Proposition~\ref{lem:functmin} to finite domains (\cref{rem:semiclfunct}) yields
	\begin{multline*}
	\limsup_{\hbar \to 0} \hbar^d \tr \log \myp{1+ e^{-\beta \myp{|i\hbar \nabla + A|^2 + V}}}\\
	\leq \frac{1}{\myp{2 \pi}^d} \iint_{\mathbb{R}^{2d}} \log \myp{1+ e^{-\beta \myp{p^2 + V\myp{x}}}} \, \mathrm{d} x \, \mathrm{d} p +\varepsilon \myp{R},		
	\end{multline*}
	where $ \varepsilon \myp{R} \to 0 $ when $ R \to \infty $. This concludes the proof in the case $V_-\in L^\infty(\mathbb{R}^d)$. We now remove this unnecessary assumption: let us consider a potential $ V $ satisfying the assumptions of \cref{prop:gcsemilimit} (possibly unbounded below). For $ K > 0 $, we take the cut off potential $ V_K = V \mathds{1}_{\myt{V \geq -K}} $ and for any $ 0 < \varepsilon < 1 $ we obtain using \cref{lem:unrestricted_oneb}
	\begin{align*}
		- \frac{1}{\beta} \tr \log \bigg( 1+ &e^{-\beta \myp{|i\hbar \nabla + A|^2 + V}} \bigg) \\
		&\qquad \geq \min_{0\leq \gamma \leq 1} \myp{\tr \myp{ \myp{1-\varepsilon} |i\hbar \nabla + A|^2 + V_K} \gamma + \frac{1}{\beta} \tr s \myp{\gamma} }\\
		&\qquad \qquad \qquad \qquad+ \min_{0\leq \gamma\leq 1} \tr \myp{\varepsilon |i\hbar \nabla + A|^2 + V-V_K}\gamma \\
		&\qquad = - \frac{1}{\beta} \tr \log \myp{1+ e^{-\beta \myp{\myp{1-\varepsilon}|i\hbar \nabla + A|^2 + V_K}}} \\
		&\qquad \qquad \qquad \qquad \qquad \qquad- \tr \myp{\varepsilon|i\hbar \nabla + A|^2 + V-V_K}_-.
	\end{align*}
	Applying the Lieb-Thirring inequality, we obtain
	\begin{equation*}
		\tr \myp{\varepsilon|i\hbar \nabla + A|^2 + V-V_K}_- \leq C \hbar^{-d} \varepsilon^{-d/2} \int_{\mathbb{R}^d} \myp{V-V_K}_-^{1+d/2} \, \mathrm{d} x.
	\end{equation*}
	This means that for any $ K $ and $ \varepsilon $
	\begin{align*}
		\limsup_{\hbar \to 0} \hbar^d \tr \log &\myp{1+ e^{-\beta \myp{|i\hbar \nabla + A|^2 + V}}} \\
		&\leq \frac{1}{\myp{2 \pi}^d} \iint_{\mathbb{R}^{2d}} \log \myp{1+ e^{-\beta \myp{\myp{1-\varepsilon} p^2 + V_K \myp{x}}}} \, \mathrm{d} x \, \mathrm{d} p \\
		&\qquad \qquad \qquad\qquad \qquad+ \varepsilon^{-d/2} C \int_{\mathbb{R}^d} \myp{V-V_K}_-^{1+d/2} \, \mathrm{d} x.
	\end{align*}
	First taking $ K \to \infty $ and afterwards $ \varepsilon \to 0 $, the result follows using the monotone convergence theorem.
\end{proof}

\section{Proof of Theorem~\ref{theo:TL_Can_Sc} in the general case}
\label{sec:proof_interacting}

In this section we deal with the interacting case $w\neq0$. We first focus on the proof of Theorem~\ref{theo:TL_Can_Sc} (mean-field limit) before proving Theorem~\ref{theo:dilute_limit} (dilute limit).

\subsection{Convergence of the energy in the mean-field limit $\eta=0$}
Here we prove \eqref{eq:limit} in the case of general $ w \in L^{1+d/2} \myp{\mathbb{R}^d} + L_{\varepsilon}^{\infty} \myp{\mathbb{R}^d} $. The upper bound on the canonical energy follows immediately from the trial states constructed in Proposition~\ref{prop:trial_states}, so we concentrate on proving the lower bound. This is the content of the following proposition.

\begin{prop}\label{prop:lower_bound_canonical_by_sc}
	Let $ \beta_0, \rho > 0 $, $ V \in L_{\loc}^{1+{d}/{2}} \myp{\mathbb{R}^d} $ such that $ V \myp{x} \to \infty $ when $ \left\lvert x \right\rvert \to \infty $ and $ \int e^{-\beta_0V_+\myp{x}} \, \mathrm{d} x < \infty $.
	Furthermore, let $ |A|^2,w \in L^{1+{d}/{2}} \myp{\mathbb{R}^d} + L_{\varepsilon}^{\infty} \myp{\mathbb{R}^d} $, $w$ be even and satisfy $\widehat{w} \geq 0$.
	Then we have
	\begin{equation*}
		\liminf_{\substack{N\to\infty \\ \hbar^d N \to \rho}} \hbar^d \ecan(\hbar,N) \geq \esc(\rho).
	\end{equation*}
\end{prop}

\begin{proof}
	The main idea of the proof is to replace $ w $ by an effective one-body potential, and then use the lower bound in the non-interacting case.	
	
	We begin by regularizing the interaction potential: let $\varphi \in C^\infty_c(\mathbb{R}^d)$ even and real-valued, define $\chi = \varphi \ast \varphi$ and $w_\varepsilon = w \ast \chi_\varepsilon$ with $\chi_\varepsilon = \varepsilon^{-d} \chi(\varepsilon^{-1} \cdot)$ for $\varepsilon >0$. Note that $\widehat{w_\varepsilon} \geq 0$. Moreover, if $\alpha >0$ and $w=w_1 + w_2$ with $w_1 \in L^{1+\frac{d}{2}}(\mathbb{R}^d)$ and $\|w_2\|_{L^\infty(\mathbb{R}^d)}\leq \alpha$ then $w_{1,\varepsilon} := w_1 \ast \chi_{\varepsilon}$ satisfies $\widehat{w_{1,\varepsilon} } \in L^1(\mathbb{R}^d)$ and $w_{2,\varepsilon} := w_2 \ast \chi_{\varepsilon}$ satisfies $\|w_{2,\varepsilon}\|_{L^\infty(\mathbb{R}^d)} \leq \alpha $. Then, using the Lieb-Thirring inequality, we can replace $w$ by $w_\varepsilon$ up to an error of order ${\|w_1-w_{1,\varepsilon}\|_{L^{1+d/2}(\mathbb{R}^d)}} + C \alpha$, see for instance \cite[Lemma $ 3.4 $]{FouLewSol-18}. It remains to let $\varepsilon$ tend to zero and then let $\alpha$ tend to zero. We therefore assume for the rest of the proof that $w$ satisfies $\widehat{w}\in L^1(\mathbb{R}^d)$.
	
	Now, with $ 0\leq \widehat{w} \in L^1(\mathbb{R}^d) $, it is classical that we can bound $ w $ from below by a one-body potential, see, e.g.,~\cite[Lem.~3.6]{FouLewSol-18}.
	More precisely, we have for all $x_1,\hdots, x_N \in \mathbb{R}^d$ and $\varphi \in C^\infty_c(\mathbb{R}^d)$
	\begin{equation*}
	\int_{\mathbb{R}^d} \widehat{w} \left| \widehat{ \sum_{i=1}^N \delta_{x_i} - \varphi }\right|^2 \geq 0,
	\end{equation*}
	which after expanding is the same as
	\begin{equation}
	\label{ineq:w_hat_pos}
	\sum_{1\leq i<j \leq N} w(x_i -x_j) \geq \sum_{i=1}^N w\ast \varphi (x_i) - \frac{1}{2} \int_{\mathbb{R}^d} (\varphi \ast w) \varphi - \frac{N}{2} w(0).
	\end{equation}
	Let $m_0$ be the minimizer of the semiclassical problem with density $\rho$, whose existence is guaranteed by \cref{theo:min_Vlasov}.
	For any $ N $-body trial state $ \Gamma $ we obtain from (\ref{ineq:w_hat_pos})
	\begin{align*}
		\tr H_{N,\hbar} \Gamma 
		&\geq \tr \left((i\hbar \nabla + A(x))^2 + V(x) + \rho^{-1}w \ast \rho_{m_0}(x) \right) \Gamma^{(1)} \\
		&\qquad- \frac{N}{2\rho^2} \int_{\mathbb{R}^d} (\rho_{m_0} \ast w) \rho_{m_0} - \frac{1}{2} w(0),
	\end{align*}
	where $\Gamma^{(1)}$ is the $1$-particle reduced density matrix of $\Gamma$.
	Let $ \Musc(\rho) $ be the chemical potential corresponding to the minimizer $ m_0 $ and define $ V^{\mathrm{eff}} = V + \rho^{-1}w \ast \rho_{m_0}(x) - \Musc(\rho)$.
	Denoting by $ e_{\mathrm{Can}}^{\beta,\mathrm{eff}} \myp{\hbar,N} $ the minimum of the canonical energy with potential $ V^{\mathrm{eff}} $ and with no interaction, we obtain using the convergence shown for the non-interacting case in Section~\ref{sec:proof_main_thm},
	\begin{align*}
		\hbar^d \ecan \myp{\hbar,N}
		&\geq \hbar^d e_{\mathrm{Can}}^{\beta,\mathrm{eff}} \myp{\hbar,N} - \frac{\hbar^d N}{2\rho^2} \int_{\mathbb{R}^d} (\rho_{m_0} \ast w) \rho_{m_0} + \Musc(\rho) \hbar^d N \\
		&\!\!\underset{\substack{N\to\infty\\\hbar^d N \to \rho}}{\longrightarrow} -\frac{1}{\beta(2\pi)^d} \iint_{\mathbb{R}^{2d}} \log(1+ e^{-\beta(p^2+ V^{\mathrm{eff}}(x))}) \, \mathrm{d} x \, \mathrm{d} p \\
		&\qquad\qquad - \frac{1}{2\rho} \int_{\mathbb{R}^d} (\rho_{m_0} \ast w) \rho_{m_0} + \Musc(\rho) \rho \\
		&= \esc(\rho),
	\end{align*}
	where the last equality is due to \cref{theo:min_Vlasov}.
	This concludes the proof of the convergence of energy in \cref{theo:TL_Can_Sc}.
\end{proof}


\subsection{Convergence of the energy in the dilute limit $\eta>0$}

Here we prove the convergence of the energy in Theorem~\ref{theo:dilute_limit} where $\eta > 0$. We first state a lemma about the regularity of the minimizers of (\ref{eq:def_min_TF}) when the interaction has a Dirac component. It will be needed in the proof of the convergence of the energy in \cref{theo:dilute_limit} below.

\begin{lem}
	\label{lem:rho_regularity}
	Let $ \beta,a,\rho > 0 $, let $A,V$ satisfy the assumptions of \cref{theo:min_Vlasov}, let $w = a \delta_0$ for some $a>0$. If $ m \in L^{1}(\mathbb{R}^{2d}) $ satisfies the non-linear equation (\ref{eq:m0def}), then $ \rho_m \in L^{1+d/2} \myp{\mathbb{R}^d} $.
\end{lem}
\begin{proof}
	For simplicity and without loss of generality, we assume that $ a=\rho=1$, $\mu = 0$ and we take $w = \delta_0$ and $A=0$.
	Since $ \rho_m \in L^1 \myp{\mathbb{R}^d} $, it is sufficient to show that $ \rho_m \mathds{1}_{\{\rho_m\myp{x}\geq 1\}} $ is in $ L^{1+d/2} \myp{\mathbb{R}^d} $.
	Recalling that $ m $ satisfies the equation
	\begin{equation}
		\label{eq:m0def2}
		m \myp{x,p} = \frac{1}{1+ e^{\beta \myp{p^2 + V\myp{x} + \rho_m \myp{x}}}},
	\end{equation}
	we immediately have
	\begin{equation*}
		\rho_m \myp{x} \leq \frac{e^{-\beta \myp{V\myp{x}+\rho_m \myp{x} }} }{\myp{2 \pi}^d} \int_{\mathbb{R}^d} e^{-\beta p^2} \, \mathrm{d} p
		= C_{d,\beta} e^{-\beta \myp{V\myp{x}+\rho_m \myp{x} }},
	\end{equation*}
	implying that
	\begin{equation*}
		\rho_{m}(x)e^{\beta \rho_m \myp{x} } \leq C_{d,\beta} e^{\beta V_- \myp{x}}.
	\end{equation*}
	Hence
	\begin{equation*}
	\rho_{m}\mathds{1}_{\left\{\rho_m \geq 1\right\}} 
	\leq \left(V_- + \log C_{d,\beta} \right) \mathds{1}_{\left\{\rho_m \geq 1\right\}} 
	\in L^{1+\frac{d}{2}}(\mathbb{R}^d),
	\end{equation*}
	since $ V_- \in L^{1+\frac{d}{2}} \myp{\mathbb{R}^d} $ and $ \{\rho_m \geq 1 \} $ has finite measure by Markov's inequality.
\end{proof}

\begin{rem}
If $w=0$ then $\rho_{m}$ behaves like $V_-^{d/2}$, it can be seen by doing the same computation as in (\ref{eq:computation_rho}). Therefore, without other assumptions than $w \in L^{1+\frac{d}{2}}(\mathbb{R}^d)$, we cannot expect more from $\rho_{m}$.
\end{rem}

\subsubsection{Case $ 0 < \eta < 1/d $}
We assume that $ 0 < d \eta < 1 $ and take $ w \in L^1 \myp{\mathbb{R}^d} $ with $ 0 \leq \widehat{w} \in L^1 \myp{\mathbb{R}^d} $.
Take $ w_N = N^{d \eta} w \myp{N^{\eta} \cdot} $ and consider the canonical model with this interaction. Denoting $ a = \int_{\mathbb{R}^d} w \myp{x} \, \mathrm{d} x $, \cref{prop:trial_states} implies that
\begin{equation*}
	\limsup_{\substack{N \to \infty \\ \hbar^d N \to \rho}} \hbar^d \ecan \myp{N,\hbar}
	\leq e_{\mathrm{Vla}}^{\beta,w=a \delta_0} \myp{\rho}.
\end{equation*}
To show the lower bound, we follow the argument of \cref{prop:lower_bound_canonical_by_sc}.
Denote by $ m_0 $ the minimizer of the Vlasov functional with the delta interaction $ a \delta_0 $, and let $ \Gamma_N $ be the Gibbs state minimizing the canonical free energy functional.
Applying \eqref{ineq:w_hat_pos} with $ \varphi = \frac{N}{\rho} \rho_{m_0} $, we obtain
\begin{align}
	\tr H_{N,\hbar} \Gamma_N 
	&\geq \tr \myp{ \myp{i \hbar \nabla + A}^2 + \Veff} \Gamma_N^{\myp{1}} + \frac{1}{\rho} \tr \myp{ w_N \ast \rho_{m_0} - a \rho_{m_0}} \Gamma_N^{\myp{1}} \nonumber \\
	&\qquad - \frac{N}{2 \rho^2} \int_{\mathbb{R}^d} \myp{\rho_{m_0} \ast w_N} \rho_{m_0} + \Musc^{w = a \delta_0} \myp{\rho} N + o \myp{\hbar^{-d}},
\label{eq:dilutebound}
\end{align}
where $ \Veff = V + \frac{a}{\rho} \rho_{m_0} - \Musc^{w = a \delta_0} \myp{\rho} $.
Here, by H\"older's inequality, we have
\begin{align*}
	\hbar^d \tr (w_N \ast \rho_{m_0} - a \rho_{m_0})& \Gamma_N^{\myp{1}} \\
	&= \hbar^d \int_{\mathbb{R}^d} \left( w_N \ast \rho_{m_0} - a \rho_{m_0}\right) \rho_{\Gamma_N^{\myp{1}}} \\
	&\leq \left\| \hbar^d \rho_{\Gamma_N^{\myp{1}}} \right\|_{L^{1+2/d}(\mathbb{R}^d)} \left\| w_N \ast \rho_{m_0} - a \rho_{m_0} \right\|_{L^{1+d/2}(\mathbb{R}^d)},
\end{align*}
which tends to $0$ since $ \| \hbar^d \rho_{\Gamma_N^{\myp{1}}}\|_{L^{1+{2}/{d}}(\mathbb{R}^d)} $ is bounded, by the Lieb-Thirring inequality, and since $ \rho_{m_0} \in L^{1+\frac{d}{2}} \myp{\mathbb{R}^d} $ by \cref{lem:rho_regularity}. Finally we have,
\begin{equation*}
	\int_{\mathbb{R}^d} \myp{\rho_{m_0} \ast w_N } \rho_{m_0} \longrightarrow a \int_{\mathbb{R}^d} \rho_{m_0}^2.
\end{equation*}
Hence, continuing from \eqref{eq:dilutebound}, we conclude that
\begin{align*}
	\liminf_{\substack{N \to \infty \\ \hbar^d N \to \rho}} \hbar^d \ecan \myp{N,\hbar}
	&\geq - \frac{1}{\myp{2 \pi}^d \beta} \iint_{\mathbb{R}^{2d}} \log \myp{1+ e^{-\beta \myp{p^2 + \Veff \myp{x}}}} \, \mathrm{d} x \, \mathrm{d} p \\
	&\qquad \qquad \qquad \qquad \qquad \quad + \Musc^{w = a \delta_0} \myp{\rho} \rho -\frac{a}{2 \rho} \int_{\mathbb{R}^d} \rho_{m_0}^2 \\
	&= e_{\mathrm{Vla}}^{\beta,w=a \delta_0} \myp{\rho}.
\end{align*}

\subsubsection{Case $ \eta > 1/d $}
Here we treat the dilute limit. Assume that $d\geq 3$, $ 0 \leq w \in L^1 \myp{\mathbb{R}^d} $, and that $ w $ is compactly supported. Then, since $ w \geq 0 $, we have the immediate lower bound
\begin{equation*}
	\liminf_{\substack{N \to \infty \\ \hbar^d N \to \rho}} \hbar^d \ecan \myp{N,\hbar}
	\geq \liminf_{\substack{N \to \infty \\ \hbar^d N \to \rho}} \hbar^d e_{\mathrm{Can}}^{\beta,w = 0} \myp{N,\hbar}
	= e_{\mathrm{Vla}}^{\beta,w=0} \myp{\rho}.
\end{equation*}
On the other hand, it follows from \cref{prop:trial_states} that we also have the corresponding upper bound, so
\begin{equation*}
	\lim_{\substack{N \to \infty \\ \hbar^d N \to \rho}} \hbar^d \ecan \myp{N,\hbar} = e_{\mathrm{Vla}}^{\beta,w=0} \myp{\rho}.
\end{equation*}
This finishes the proof of the convergence of the energy in the dilute limit.


\subsection{Convergence of states}\label{sec:CV_states}


\subsubsection{Strong convergence of the one-particle Husimi and Wigner measures}
Here we concentrate on proving the limits~\eqref{eq:limit_1PDM_a} an~\eqref{eq:limit_1PDM_b} for the one-particle Husimi measure and the associated density. 
We start by briefly recalling the definitions.

For $f\in L^2(\mathbb{R}^d)$ a normalized, real-valued function and $(x,p)\in \mathbb{R}^{2d}$, $\hbar >0$, we define $f_{x,p}^\hbar(y) = \hbar^{-d/4} f((x-y)/\hbar^{1/2})e^{ip\cdot y / \hbar}$ and denote by $P_{x,p}^\hbar = \ket{f_{x,p}^\hbar}\bra{f_{x,p}^\hbar}$ the orthogonal projection onto $f_{x,p}^\hbar$. For $k\geq 1$, we introduce the $k$-particle Husimi measure of a state $\Gamma$
\begin{equation*}
m^{(k)}_{f,\Gamma}(x_1,p_1,...,x_k,p_k) = \frac{N!}{(N-k)!} \tr\left( P_{x_1,p_1}^\hbar \otimes \cdots  \otimes P_{x_k,p_k}^\hbar \otimes \mathds{1}_{N-k} \Gamma\right),
\end{equation*}
for $x_1,p_1,...,x_k,p_k \in \mathbb{R}^{2dk}$. See~\cite{FouLewSol-18} for alternative formulas of $m^{(k)}_{f,\Gamma}$. We also recall the definition of the Wigner measure,
\begin{equation*}
\begin{gathered}
\mathcal{W}_{\Gamma}^{(k)}(x_1,..., p_k) = \int_{\mathbb{R}^{dk}} \int_{\mathbb{R}^{d(N-k)}} e^{-i\sum_{\ell=1}^k p_\ell \cdot y_\ell} \times \\
\times \Gamma(x_1  + \hbar y_1/2,...,x_k  + \hbar y_k /2, z_{k+1},...,z_N) \, dy_1...dy_k dz_{k+1}...dz_{N},
\end{gathered}
\end{equation*}
for $x_1,p_1,...,x_k,p_k \in \mathbb{R}^{2dk}$, where $\Gamma(\cdot,\cdot)$ is the kernel of the operator $\Gamma$.

Using \cite[Theorem 2.7]{FouLewSol-18} and the fact that the Husimi measures are bounded both in the $x$ and $p$ variables, we obtain the existence of a Borel probability measure $\mathcal{P}$ on 
$$\mathcal{S} =\left\{ \mu \in L^1(\mathbb{R}^{2d}),\quad 0\leq \mu \leq 1,\quad \int_{\mathbb{R}^{2d}} \mu = \rho \right\}$$ 
such that, up to a subsequence, we have
\begin{equation*}
\int_{\mathbb{R}^{2dk}} m_{f,\Gamma_N}^{(k)} \varphi  \to \int_{\mathcal{S}}\left(\int_{\mathbb{R}^{2dk}} m^{\otimes k} \varphi \right)d\mathcal{P}(m),
\end{equation*}
for any $\varphi \in L^{1}(\mathbb{R}^{2dk}) + L^{\infty}(\mathbb{R}^{2dk})$ and similarly for the Wigner measures. There is no loss of mass in the limit due to the confining potential $V$. Our goal is to show that $\mathcal{P}=\delta_{m_0}$, where $m_0$ is the Vlasov minimizer from Theorem~\ref{theo:min_Vlasov}.

We begin with the case $\eta = 0$. Using coherent states, the tightness of $(m_{f,\Gamma_N}^{(1)})_N$ and a finite volume approximation we obtain
\begin{align}
&\lim_{\substack{N_j\to\infty \\ \hbar^d N_j \to \rho}} \hbar^d \ecan(\hbar,N_j)  \geq  \frac{1}{(2\pi)^d}\int_{\mathcal{S}}\left(\int_{\mathbb{R}^{2dk}}(p^2 + V(x)) m(x,p)  \right)d\mathcal{P}(m)\nn  \\
&\quad +  \frac{1}{2\rho}\int_{\mathcal{S}}\left(\int_{\mathbb{R}^{2dk}} (w\ast \rho_m) \rho_m \right)d\mathcal{P}(m) + \frac{1}{(2\pi)^d} \int_{\mathbb{R}^{2d}} s\left( \int_{\mathcal{S}} m\, d\mathcal{P}(m)\right).\label{eq:pass_to_the_limit}
\end{align}
The lower semi-continuity of the entropy term can be justified as in the proof of \cref{lem:functminexist}. The case $0<\eta < 1/d$ can be adapted using (\ref{ineq:w_hat_pos}) with $\varphi= N \rho_{m_0}$ and the case $\eta > 1/d$ is even easier since the interaction is assumed non-negative and can therefore be dropped. 

If we denote $\overline{m} =\int_{\mathcal{S}} m\, d\mathcal{P}(m)$, the right side of~\eqref{eq:pass_to_the_limit} is not exactly $\Ecan(\overline{m})$ because of the interaction term. In the case $0 \leq \eta < 1/d$ we assumed $\widehat{w}\geq 0$, hence the following inequality follows from convexity: 
\begin{equation}
\int_{\mathcal{S}}\left(\int_{\mathbb{R}^{2dk}} w\ast \rho_m \rho_m \right)d\mathcal{P}(m) \geq \int_{\mathbb{R}^{2d}} w\ast \rho_{\overline{m}}\rho_{\overline{m}}.
\label{eq:convex_interaction}
\end{equation}
The case $1/d < \eta $ is immediate since we assumed $w\geq 0$ and the limiting energy has no interaction term.
Gathering the above inequalities we have
\begin{equation*}
\lim_{\substack{N_j\to\infty \\ \hbar^d N_j \to \rho}} \hbar^d \ecan(\hbar,N_j) \geq \mathcal{E}_{\rm Vla}^{\beta,\rho,\bullet}(\overline{m})\geq e_{\mathrm{Vla}}^{\beta,\bullet} \myp{\rho},
\end{equation*}
where $\mathcal{E}_{\rm Vla}^{\beta,\rho,\bullet}$ and $ e_{\mathrm{Vla}}^{\beta,\bullet} \myp{\rho}$ are the appropriate limiting functional and energy: i.e. $\bullet= w$ if $\eta = 0$, $\bullet= (\int_{\mathbb{R}^d} w)\delta_0$ if $0<d \eta < 1$ and $\bullet = 0$ if $d \eta \geq 1$ and $d \geq 3$. Now, the equality in this series of inequalities forces $\bar m$ to be equal to $m_0$. And since this limit does not depend on the subsequence we have taken, we conclude that the whole sequence $m^{(1)}_{f,\Gamma_N}$ converges weakly to $m_0$, and similarly for the Wigner measure.

Note that, when $\widehat{w}>0$ and $0<d \eta < 1$, the equality in~\eqref{eq:convex_interaction} gives that $\mathcal{P}$ is concentrated on functions $m$ which all share the same density $\rho_{m_0}$, by strict convexity. But this is the only information that we have obtained so far on $\mathcal{P}$. If the conjectured entropy inequality~\eqref{eq:Fatou_entropy} was valid, then we would conclude immediately that $\mathcal{P}=\delta_{m_0}$. Since we do not have this inequality, we will have to go back later to the proof that $\mathcal{P}=\delta_{m_0}$. 

So far the convergence of $m^{(1)}_{f,\Gamma_N}$ is only weak but it can be improved using the (one-particle) entropy. Going back to the previous estimates we now have 
\begin{equation}
\hbar^d \ecan(\hbar,N) 
= e_{\mathrm{Vla}}^{\beta,\bullet} \myp{\rho}+ \frac{1}{(2\pi)^d \beta}\iint_{\mathbb{R}^{2d}} (s(m_{f,\Gamma_N}^{(1)}) - s(m_0))+ o(1) \label{eq:entropy_terms}
\end{equation}
As before we denote by $e_{\mathrm{Vla}}^{\beta,\bullet} \myp{\rho}$ the appropriate limiting energy, depending on the choice of $\eta$. Recall that in the case $\eta > 1/d$, the interaction potential is assumed to be non negative, so the interaction term is just dropped. We now focus on the second term in (\ref{eq:entropy_terms}). Let us remark that
\begin{align*}
&s(m_{f,\Gamma_N}^{(1)}) - s(m_0) \\
&\ = m_{f,\Gamma_N}^{(1)} \log\left(\frac{m_{f,\Gamma_N}^{(1)}}{m_0}\right) + (1-m_{f,\Gamma_N}^{(1)}) \log\left(\frac{1-m_{f,\Gamma_N}^{(1)}}{1-m_0}\right) \\
	& \qquad \qquad \qquad \qquad \qquad \qquad \qquad \qquad \qquad \quad 
	+ (m_0 - m_{f,\Gamma_N}^{(1)}) \log\left(\frac{1-m_0}{m_0}\right) \\
	&\  \geq m_0 \log\left(\frac{m_{f,\Gamma_N}^{(1)}}{m_0}\right) +  \beta (m_0 - m_{f,\Gamma_N}^{(1)}) \left(p^2 + V+ \frac{1}{\rho} w_N\ast \rho_{m_0} - \mu + \beta^{-1}\right),
\end{align*}
where we used the expression of $m_0$ (\ref{eq:m0def}) and the pointwise inequality \\ $x \log(x/y) + (y-x) \geq 0$ for any $x,y > 0$. Integrating over $x$ and $p$, we obtain on the right side the sum of the relative von Neumann entropy of $m_{f,\Gamma_N}^{(1)}$ and $m_0$, and a term which tends to zero, due  to the weak convergence we have proven. By Pinsker's inequality and (\ref{eq:entropy_terms}) we obtain
\begin{align*}
\hbar^d \ecan(\hbar,N) - e_{\mathrm{Vla}}^{\beta,\bullet} \myp{\rho}
      \geq \frac{1}{2 (2\pi)^d \beta} \bigg( &\int_{\mathbb{R}^{2d}} |m_{f,\Gamma_N}^{(1)} - m_0|\bigg)^2 + o(1).
\end{align*}
The convergence of the energies gives the strong convergence in $L^1(\mathbb{R}^{2d})$ of $m_{f,\Gamma_N}^{(1)}$ towards the Vlasov minimizer $m_0$, hence in $L^p(\R^{2d})$ for all $1\leq p<\ii$ since the Husimi measures are bounded by 1. This automatically gives that $\rho_{m_{f,\Gamma_N}^{(1)}} \to \rho_{m_0}$ strongly in $L^1(\mathbb{R}^d)$. The weak convergence in $L^{1+2/d}(\mathbb{R}^d)$ follows from the (classical) Lieb-Thirring inequality
\begin{equation*}
\|\rho_m\|_{L^{1+d/2}(\mathbb{R}^d)} \leq C \|m\|_{L^1(\mathbb{R}^{2d},p^2dxdp)}^{\frac{d}{d+2}} \|m\|_{L^\infty(\mathbb{R}^{2d})}^{\frac{2}{d+2}}
\end{equation*}
for any $m$ in $L^1(\mathbb{R}^{2d})$.

Finally, by the Lieb-Thirring inequality $ \hbar^d \rho_{\Gamma_N}^{(1)}$ is bounded in $L^{1}(\mathbb{R}^d) \cap L^{1+d/2}(\mathbb{R}^d)$, hence this sequence is weakly precompact in those spaces. On the other hand, for any $\varphi \in C^{\infty}_c(\mathbb{R}^d)$ we have by~\cite[Lemma 2.4]{FouLewSol-18}
\begin{equation*}
\int_{\mathbb{R}^d} \rho_{m_{f,\Gamma_N}^{(1)}} \varphi =  \int_{\mathbb{R}^d} \hbar^d \rho_{\Gamma_N}^{(1)} \varphi \ast |f^{\hbar}|^2.
\end{equation*}
Let $\widetilde{\rho}$ be an accumulation point for $ \hbar^d\rho_{\Gamma_N}^{(1)}$. By passing to the limit in both sides we obtain
\begin{equation*}
\int_{\mathbb{R}^d}  \rho_{m_0} \varphi =\int_{\mathbb{R}^d} \widetilde{\rho} \varphi.
\end{equation*}
The test function $\varphi$ being arbitrary, we conclude that $\hbar^d \rho_{\Gamma_N}^{(1)}$ has a single accumulation point and therefore converges weakly in $L^1(\mathbb{R}^d) \cap L^{1+d/2}(\mathbb{R}^d)$ towards $\rho_{m_0}$.

\subsubsection{Weak convergence of the $k$-particle Husimi and Wigner measures}

At this point we have proved the strong convergence of $m^{(1)}_{f,\Gamma_N}$ towards $m_0$ in $L^p(\R^{2d})$ for all $1\leq p<\ii$. Our argument works for any sequence of approximate Gibbs states $(\Gamma_N)$ in the sense that 
$$\mathcal{E}_{\mathrm{Can}}^{N,\hbar} \myp{\Gamma_N}=e_{\mathrm{Can}}^{\beta} \myp{\hbar,N} +o(N).$$ 
Here we discuss the weak convergence of the higher order Husimi functions. This is not an easy fact in the canonical ensemble case. For instance, when $w\equiv0$ one can use Wick's formula in the grand canonical case but there is no such formula in the canonical ensemble~\cite{Schonhammer-17,GraMajSchTex-18}. Here we will use a Feynman-Hellmann-type argument, which forces us to consider the exact Gibbs state, and not only an approximate equilibrium state. We will come back to approximate Gibbs states at the end of the proof but our argument will require that they approach the right energy with an error of order $o(1)$ instead of $o(N)$. 

In order to access the two-particle Husimi function, the usual Feynman-Hellmann argument is to perturb the $N$-body Hamiltonian by a positive two-body term of order $N$, multiplied by a small parameter $\eps$. This modifies the effective Vlasov energy and, after taking the limit, one then look at the derivative at $\eps=0$. The problem here is to control negative values of $\eps$. For atoms one can use the strong repulsion at the origin of the Coulomb interaction to control a negative two-body term, as was done in~\cite{NarThi-81}.\footnote{After inspection one sees that the argument used in~\cite{NarThi-81} works under the condition that $\widehat{w}(p)\geq a|p|^{-a}$ for some $a>0$ for large $p$. Not all interaction potentials can therefore be covered.} For a general interaction or even when $w\equiv0$, such an argument fails. Another difficulty is the need to re-prove the existence of the limit with the perturbation, since in the canonical ensemble trial states are not so easy to construct. 

We follow a different route and use instead an argument inspired of a new technique recently introduced in~\cite{LewNamRou-18c}. The idea is to perturb the energy by a one body term of order $1$. This will not modify the leading order in the limit and will force us to look at the next order. Since we are only interested in deviations in $\eps$, the existence of the limit for the one-particle Husimi measure will help us to identify the deviation. Then, in order to access the two-body Husimi measure, we look at the second derivative at $\eps=0$ instead of the first derivative.

Let us detail the argument. Let $b\in C^\ii_c(\R^d\times \R^d,\R_+)$ be a non-negative function on the phase space and introduce its coherent state quantization
$$B_\hbar:=\frac{1}{(2\pi)^d}\iint_{\R^d\times\R^d}b(x,p)\,P_{x,p}^\hbar\,dx\,dp,$$
where we recall that $P_{x,p}^\hbar=\ket{f_{x,p}^\hbar}\bra{f_{x,p}^\hbar}$ is the orthogonal projection onto $f_{x,p}^\hbar$. We then consider the one-particle operator
\begin{equation}
B_{N,\hbar}:=\sum_{j=1}^N(B_\hbar)_j
\label{eq:def_B_N}
\end{equation}
in the $N$-particle space. Note that $B_\hbar$ is a bounded self-adjoint operator with 
$$0\leq B_\hbar\leq \|b\|_{L^\ii(\R^{2d})}\hbar^d$$
due to the coherent state representation~\eqref{eq:coherent_state_repres} and that it is trace-class with
\begin{equation}
\tr(B_\hbar)\leq \frac{1}{(2\pi)^d}\iint_{\R^d\times\R^d}b(x,p)\,dx\,dp.
\label{eq:estim_trace_B}
\end{equation}
In particular $B_{N,\hbar}$ is bounded uniformly in $N$, with 
\begin{equation}
\norm{B_{N,\hbar}}\leq \min\bigg(N\hbar^d\|b\|_{L^\ii(\R^{2d})},(2\pi)^{-d}\|b\|_{L^1(\R^{2d})}\bigg).
\label{eq:estim_norm_B}
\end{equation}
This is because 
\begin{equation}
 \norm{\sum_{j=1}^NC_j}\leq \norm{C}_{\gS^1}
 \label{eq:norm_many-body}
\end{equation}
in the fermionic $N$-particle space.
We introduce the perturbed Hamiltonian
$$H_{N,\hbar}(\eps):=H_{N,\hbar}+\eps B_{N,\hbar},$$
for $\eps\in\R$. The perturbation is uniformly bounded, hence will not affect the limit $N\to\ii$ for fixed $\eps$. More precisely, let us call 
$$\Gamma_{N,\hbar,\beta}(\eps):=\frac{e^{-\beta H_{N,\hbar}(\eps)}}{\tr e^{-\beta H_{N,\hbar}(\eps)}}$$
the associated Gibbs state and
$$F_{N,\hbar,\beta}(\eps):=-\frac{\log\tr(e^{-\beta H_{N,\hbar}(\eps)})}{\beta}$$
the corresponding free energy. Everywhere we assume that $\hbar N^{1/d}\to\rho$ and $\beta>0$ is fixed. By plugging $\Gamma_{N,\hbar,\beta}(\eps)$ into the variational principle at $\eps=0$ and conversely, we obtain immediately that
\begin{multline}
F_{N,\hbar,\beta}(0)+\frac{\eps}{(2\pi)^d}\iint_{\R^{2d}} b\, m^{(1)}_{f,\Gamma_{N,\hbar,\beta}(\eps)}\\
\leq F_{N,\hbar,\beta}(\eps)\leq F_{N,\hbar,\beta}(0)+\frac{\eps}{(2\pi)^d}\iint_{\R^{2d}} b\, m^{(1)}_{f,\Gamma_{N,\hbar,\beta}(0)}. 
\label{eq:compare_eps}
\end{multline}
We have used here that $\tr(B_\hbar \Gamma^{(1)})=(2\pi)^{-d}\iint_{\R^{2d}}b\,m^{(1)}_{f,\Gamma}$ for all states $\Gamma$. 
Since $0\leq m^{(1)}_{f,\Gamma}\leq1$, this proves that 
$$F_{N,\hbar,\beta}(\eps)=\mathcal{E}_{\mathrm{Can}}^{N,\hbar} \myp{\Gamma_{N,\hbar,\beta}(\eps)}=e_{\mathrm{Can}}^{\beta} \myp{\hbar,N} +O(\eps)$$
where $O(\eps)$ is even uniform in $N$. Hence from the analysis in the previous section, we deduce immediately that 
$$m^{(1)}_{\Gamma_{N,\hbar,\beta}(\eps)}\longrightarrow m_0$$
strongly in $L^1(\R^{2d})$ for any fixed $\eps$. Going back to~\eqref{eq:compare_eps} we infer that
$$F_{N,\hbar,\beta}(\eps)= F_{N,\hbar,\beta}(0)+\frac{\eps}{(2\pi)^d}\iint_{\R^{2d}} b m_0+o(1).$$
A different way to state the same limit is
\begin{equation}
f_N(\eps):=\frac{\tr e^{-\beta H_{N,\hbar}-\beta \eps B_{N,\hbar}}}{\tr e^{-\beta H_{N,\hbar}}}\longrightarrow \exp\left(-\frac{\eps\beta}{(2\pi)^d}\iint_{\R^{2d}} b m_0\right).
\label{eq:limit_f_N}
\end{equation}
It turns out that the so-defined function $f_N$ is $C^\ii$ on $\R$ (even real-analytic) with all its derivatives locally uniformly bounded in $N$. This follows from the following general fact.

\begin{lem}
Let $A$ be a self-adjoint operator such that $\tr(e^A)<\ii$ and let $B$ be a bounded self-adjoint operator, on a Hilbert space $\mathfrak{H}$. Then the function 
$$\eps\in\R\mapsto \frac{\tr(e^{A+\eps B})}{\tr (e^{A})}$$
is $C^\ii$ and its derivatives are bounded by
$$\left|\frac{d^k}{d\eps^k} \frac{\tr(e^{A+\eps B})}{\tr (e^{A})}\right|\leq \|B\|^k \frac{\tr(e^{A+\eps B})}{\tr (e^{A})}\leq \|B\|^ke^{|\eps|\|B\|}$$
for $k\geq0$. 
\end{lem}

\begin{proof}
Note that $\tr(e^{A+\eps B})\leq e^{\eps\|B\|}\tr(e^{A})$ since $A+\eps B\leq A+|\eps|\|B\|$. We have for the first derivative
$$\frac{d}{d\eps} \frac{\tr(e^{A+\eps B})}{\tr (e^{A})}=\frac{\tr(Be^{A+\eps B})}{\tr (e^{A})}$$
which is then clearly bounded by $\|B\|$. The second derivative is given by Duhamel's formula
\begin{equation}
\frac{d^2}{d\eps^2} \frac{\tr(e^{A+\eps B})}{\tr (e^{A})}=\int_0^1 \frac{\tr(Be^{t(A+\eps B)}Be^{(1-t)(A+\eps B)})}{\tr (e^{A})}\,dt
\label{eq:formula_2nd_derivative_exp}
\end{equation}
and we have by H\"older's inequality in Schatten spaces
\begin{align*}
\left|\tr(Be^{t(A+\eps B)}Be^{(1-t)(A+\eps B)})\right|&\leq \|B\|^2\norm{e^{t(A+\eps B)}}_{\mathfrak{S}^{\frac1t}}\norm{e^{(1-t)(A+\eps B)}}_{\mathfrak{S}^{\frac1{1-t}}}\\
&=\|B\|^2\tr(e^{A+\eps B})\leq \|B\|^2e^{|\eps|\|B\|}\tr(e^{A}),
\end{align*}
as claimed.
The argument is the same for the higher order derivatives.  The function is indeed real-analytic on $\R$ but this fact is not needed in our argument.
\end{proof}

Since in $B_{N,\hbar}$ is bounded uniformly in $N$ and $\hbar$, we conclude from the lemma that $f_N$ is bounded in $W^{k,\ii}_{\rm loc}$ for all $k$. This implies that $f_N^{(k)}$ converges locally uniformly to the $k$th derivative of the right side of~\eqref{eq:limit_f_N} for all $k$. In particular, we have
\begin{equation}
 f_N''(0)\longrightarrow\left(\frac{\beta}{(2\pi)^d}\iint_{\R^{2d}} b m_0\right)^2.
 \label{eq:value_limit_f_N_dd_1}
\end{equation}
On the other hand, we can compute the second derivative $f_N''(0)$ explicitly, using~\eqref{eq:formula_2nd_derivative_exp}:
\begin{equation}
 f_N''(0)=\beta^2\int_0^1 \frac{\tr\left(B_{N,\hbar}\,e^{-t\beta H_{N,\hbar}}B_{N,\hbar}\,e^{-(1-t)\beta H_{N,\hbar}}\right)}{\tr (e^{-\beta H_{N,\hbar}})}\,dt.
 \label{eq:formula_f_N_dd}
\end{equation}
We claim that this indeed behaves as
\begin{equation}
 f_N''(0)=\frac{\beta^2}{(2\pi)^{2d}}\iint_{\R^{4d}} b\otimes b \,m^{(2)}_{f,\Gamma_{N,\hbar,\beta}}+o(1)
 \label{eq:to_be_proven}
\end{equation}
and first explain why this is useful before justifying~\eqref{eq:to_be_proven}. From the weak convergence of $m^{(2)}_{f,\Gamma_{N,\hbar,\beta}}$ mentioned in the previous section, we obtain
\begin{equation}
 \lim_{\substack{N\to\ii\\ N\hbar^d\to\rho}}f_N''(0)=\frac{\beta^2}{(2\pi)^{2d}}\int_{\mathcal{S}} \left(\int_{\R^{2d}}bm\right)^2\,d\mathcal{P}
 \label{eq:value_limit_f_N_dd_2}
\end{equation}
with the de Finetti measure $\mathcal{P}$. Comparing~\eqref{eq:value_limit_f_N_dd_1} with~\eqref{eq:value_limit_f_N_dd_2} and using $m_0=\int_{\mathcal{S}} m\,d\mathcal{P}$, we conclude  that 
$$\int_{\mathcal{S}} \left(\int_{\R^{2d}}bm\right)^2\,d\mathcal{P}(m)=\left(\int_{\mathcal{S}}\int_{\R^{2d}}bm\,d\mathcal{P}(m)\right)^2$$
for every non-negative $b\in C^\ii_c(\R^{2d})$. This proves that $\mathcal{P}=\delta_{m_0}$ as desired. The limits~\eqref{eq:cv_of_states_husimi} and~\eqref{eq:cv_of_states_wiener} then follow for all $k\geq2$. Therefore, it only remains to prove~\eqref{eq:to_be_proven}.

The idea of the proof of~\eqref{eq:to_be_proven} is simple. Since we are in a semi-classical regime, the order of the operators in the trace~\eqref{eq:formula_f_N_dd} should not matter. If we put the two $B_{N,\hbar}$ together, we obtain after a calculation
\begin{equation*}
\tr\left(\big(B_{N,\hbar}\big)^2\Gamma_{N,\hbar,\beta}\right)=\tr\left((B_\hbar)^2\Gamma_{N,\hbar}^{(1)}\right)
+\frac{1}{(2\pi)^{2d}}\iint_{\R^{4d}}b\otimes b\, m^{(2)}_{f,\Gamma_{N,\hbar,\beta}}.
\end{equation*}
The first term tends to zero since 
$$\tr\left((B_\hbar)^2\Gamma_{N,\hbar}^{(1)}\right)\leq N\norm{B_\hbar}^2\leq \|b\|_{L^\ii(\R^{2d})}^2N\hbar^{2d},$$
whereas the second term converges to $(2\pi)^{-2d}\int_{\mathcal{S}}\left(\iint_{\R^{2d}}bm\right)^2d\mathcal{P}(m)$ due to the weak convergence of $m^{(2)}_{f,\Gamma_{N,\hbar,\beta}}$. Therefore we have to compare $f_N''(0)$ with $\tr(B_{N,\hbar})^2\Gamma_{N,\hbar,\beta}$. 

In~\cite{LewNamRou-18c}, it is proven that the function
$$t\mapsto \tr\left(B_{N,\hbar}\,e^{-t\beta H_{N,\hbar}}B_{N,\hbar}\,e^{-(1-t)\beta H_{N,\hbar}}\right)$$
is convex on $[0,1]$, non-increasing on $[0,1/2]$ and non-decreasing on $[1/2,1]$. Using that the function is minimal at $t=1/2$ and above its tangent at $t=0$ provides the bound
\begin{align*}
&\tr\left(B_{N,\hbar}\,e^{-t\beta H_{N,\hbar}}B_{N,\hbar}\,e^{-(1-t)\beta H_{N,\hbar}}\right)\\
&\qquad\geq \tr\left(B_{N,\hbar}\,e^{-\frac\beta2 H_{N,\hbar}}B_{N,\hbar}\,e^{-\frac\beta2 H_{N,\hbar}}\right)\\ 
&\qquad \geq \tr\left(\big(B_{N,\hbar}\big)^2e^{-\beta H_{N,\hbar}}\right)
+\frac{\beta}{4}\tr\left(\Big[\big[H_{N,\hbar},B_{N,\hbar}\big],B_{N,\hbar}\Big]\,e^{-\beta H_{N,\hbar}}\right)
\end{align*}
for all $t\in[0;1]$, see~\cite{LewNamRou-18c}. Inserting in~\eqref{eq:formula_f_N_dd}, we find that 
\begin{equation}
f_N''(0)\geq \beta^2\tr\left(\big(B_{N,\hbar}\big)^2\Gamma_{N,\hbar,\beta}\right)+\frac{\beta^3}4\tr\left(\Big[\big[H_{N,\hbar},B_{N,\hbar}\big],B_{N,\hbar}\Big]\,\Gamma_{N,\hbar,\beta}\right).
\label{eq:estim_from_below_F_N_dd}
\end{equation}
Hence~\eqref{eq:to_be_proven} readily follows from the following result.

\begin{lem}[Convergence of the double commutator]\label{lem:commutators}
With $B_{N,\hbar}$ as in~\eqref{eq:def_B_N}, we have
\begin{equation}
\lim_{\substack{N\to\ii\\ N\hbar^d\to\rho}}\tr\left(\Big[\big[H_{N,\hbar},B_{N,\hbar}\big],B_{N,\hbar}\Big]\,\Gamma_{N,\hbar,\beta}\right)
=0.
\label{eq:commutator_to_0}
\end{equation}
\end{lem}

\begin{proof}
We have 
\begin{multline}
\Big[\big[H_{N,\hbar},B_{N,\hbar}\big],B_{N,\hbar}\Big]=\sum_{j=1}^N\big[\big[H_{1,\hbar},B_{\hbar}\big],B_\hbar\big]_j\\
+\frac{1}{N}\sum_{1\leq j\neq  k\leq N}\big[\big[w_{jk},(B_\hbar)_j\big],(B_\hbar)_j+(B_\hbar)_k\big]
\label{eq:commutator}
\end{multline}
with $H_{1,\hbar}=|i\hbar\nabla +A|^2+V$ the one-particle operator and $w_{jk}$ the multiplication operator by $w_N(x_j-x_k)$. The commutators have been used to dramatically reduce the number of terms, but will not play any role anymore. We will estimate separately the terms $(B_{\hbar}H_{1,\hbar}B_\hbar)_j$, $(H_{1,\hbar}B_\hbar^2)_j$, $w_{jk}(B_\hbar)_j(B_\hbar)_{j'}$ and $(B_\hbar)_jw_{jk}(B_\hbar)_{j'}$ with $j'\in\{k,j\}$.

First we deal with the kinetic energy. For instance we can bound, by H\"older's inequality in Schatten spaces,
\begin{multline*}
\norm{B_\hbar (-\hbar^2\Delta) B_\hbar}_{\gS^1}+\norm{(-\hbar^2\Delta) (B_\hbar)^2}_{\gS^1}\\
\leq 2\norm{B_\hbar}\norm{B_\hbar}_{\gS^1}^{\frac12}\norm{(-\hbar^2\Delta) B_\hbar^{\frac12}}_{\gS^2}\leq \frac{C}{N}\norm{(-\hbar^2\Delta) B_\hbar^{\frac12}}_{\gS^2}.
\end{multline*}
We have used here our estimates~\eqref{eq:estim_trace_B} and~\eqref{eq:estim_norm_B} on the trace and norm of the non-negative operator $B_\hbar$. 
The last Hilbert-Schmidt norm is equal to
\begin{align*}
\norm{(-\hbar^2\Delta) B_\hbar^{\frac12}}_{\gS^2}^2&=\tr\big((-\hbar^2\Delta) B_\hbar(-\hbar^2\Delta) \big)\\
&=\frac{1}{(2\pi)^d}\iint_{\R^{2d}}b(x,p)\norm{\hbar^2\Delta f_{x,p}^\hbar}^2\,dx\,dp.
\end{align*}
Using that
\begin{multline*}
\hbar^2\Delta f_{x,p}^\hbar(y)=\hbar\;\hbar^{-d/4} (\Delta f)\left(\frac{x-y}{\sqrt{\hbar}}\right)e^{ip\cdot y / \hbar}-|p|^2f_{x,p}^\hbar(y)\\+2i\sqrt{\hbar}\;\hbar^{-d/4} p\cdot (\nabla f)\left(\frac{x-y}{\sqrt{\hbar}}\right)e^{ip\cdot y / \hbar},
\end{multline*}
we find that 
$$\norm{(-\hbar^2\Delta) B_\hbar^{\frac12}}_{\gS^2}\leq C\iint_{\R^{2d}}(|p|^4+\hbar|p|^2+\hbar^2)b(x,p)\,dx\,dp.$$
This is uniformly bounded since $b$ has a compact support in the phase space. 
Using~\eqref{eq:norm_many-body}, we conclude as we wanted that 
$$\tr\bigg(\Gamma_{N,\hbar,\beta}\sum_{j=1}^N\big[\big[-\hbar^2\Delta,B_{\hbar}\big],B_\hbar\big]_j\bigg)=O(N^{-1}).$$

For the potential term we have to use more information on the state $\Gamma_{N,\hbar,\beta}$. We first estimate
$$\tr\left(\Gamma^{(1)}_{N,\hbar,\beta}VB_\hbar^2\right)\leq\|B_\hbar\|^{\frac32}\norm{(\Gamma^{(1)}_{N,\hbar,\beta})^{\frac12}|V|^{\frac12}}_{\gS^2}\norm{|V|^{\frac12}B_\hbar^{\frac12}}_{\gS^2}.$$
Using the Lieb-Thirring inequality for $V_-$ and that the energy is $O(N)$ for $V_+$, we see that 
$$\norm{(\Gamma^{(1)}_{N,\hbar,\beta})^{1/2}|V|^{1/2}}^2_{\gS^2}= \tr\Gamma^{(1)}_{N,\hbar,\beta}|V|=O(N).$$ 
Hence we can deduce that 
$$\tr\left(\Gamma^{(1)}_{N,\hbar,\beta}VB_\hbar^2\right)\leq \frac{C}{N}\norm{|V|^{\frac12}B_\hbar^{\frac12}}.$$
Like for the kinetic energy, we compute the Hilbert-Schmidt norm
\begin{align}
\norm{|V|^{\frac12}B_\hbar^{\frac12}}_{\gS^2}^2&=\tr |V|^{\frac12}B_\hbar|V|^{\frac12}\nn\\
&=\frac{1}{(2\pi)^d}\iint_{\R^d\times\R^d}b(x,p)\norm{|V|^{\frac12}f_{x,p}^\hbar}^2_{L^2(\R^d)}\,dx\,dp\nn\\
&=\frac{1}{(2\pi)^d}\iint_{\R^{2d}}b(x,p)\;|V|\ast|f_{0,0}^\hbar|^2(x)\,dx\,dp\nn\\
&\leq C\int_{B_R}|V(x)|\,dx\label{eq:estim_trace_BVB}
\end{align}
where $B_R$ is a fixed large ball, chosen large enough such that ${\rm supp}(b)\subset B_{R-1}$. We are using here that $f_{0,0}^\hbar$ has compact support, hence ${\rm supp}(f_{0,0}^\hbar)\subset B_1$, for $\hbar$ small enough. Since $V\in L^1_{\rm loc}(\R^d)$ by assumption, this proves that
$$\tr\left(\Gamma^{(1)}_{N,\hbar,\beta}VB_\hbar^2\right)=O(N^{-1}).$$
The argument is similar for $\tr\left(\Gamma^{(1)}_{N,\hbar,\beta}B_\hbar^2V\right)$.
Finally, we also have 
\begin{equation}
\tr\left(\Gamma^{(1)}_{N,\hbar,\beta}B_\hbar V B_\hbar\right)\leq \|B_\hbar\|\norm{B_\hbar^{\frac12} |V|^{\frac12}}_{\gS^2}^2=O(N^{-1}) 
 \label{eq:estim_BVB}
\end{equation}
by~\eqref{eq:estim_trace_BVB}. This concludes the proof that the potential terms tend to 0. The argument is exactly the same for $|A|^2$. For $i\hbar\nabla\cdot A+A\cdot i\hbar\nabla$, we argue similarly, using that 
$$\norm{|A|B_\hbar^{\frac12}}_{\gS^2}+\norm{A\cdot (i\hbar\nabla)B_\hbar^{\frac12}}_{\gS^2}\leq C\int_{B_R}|A|$$
and
$$\norm{|i\hbar\nabla|\sqrt{\Gamma_{N,\hbar,\beta}^{(1)}}}_{\gS^2}^2=\tr(-\hbar^2\Delta)\Gamma_{N,\hbar,\beta}^{(1)}=O(N).$$

Let us finally turn to the interaction. First we look at
$$\tr\bigg(\Gamma_{N,\hbar,\beta}\frac{1}{N}\sum_{1\leq j\neq   k\leq N}(B_\hbar)_jw_{jk}(B_\hbar)_k\bigg)=(N-1)\tr\bigg(\Gamma_{N,\hbar,\beta}(B_\hbar)_1w_{12}(B_\hbar)_2\bigg)$$
and use the Cauchy-Schwarz inequality to estimate 
$$\pm (B_\hbar)_1w_{12}(B_\hbar)_2\leq (B_\hbar)_1|w_{12}|(B_\hbar)_1+(B_\hbar)_2|w_{12}|(B_\hbar)_2.$$
We look for instance at
$$(N-1)\tr\bigg(\Gamma_{N,\hbar,\beta}(B_\hbar)_2w_{12}(B_\hbar)_2\bigg)=\tr\bigg(\Gamma_{N,\hbar,\beta}\sum_{j=2}^N(B_\hbar)_jw_{1j}(B_\hbar)_j\bigg).$$
For fixed $x_1$, the operator $\sum_{j=2}^N(B_\hbar)_jw_{1j}(B_\hbar)_j$ (acting on the remaining $N-1$ variables) is estimated as in~\eqref{eq:norm_many-body} by
\begin{equation*}
\norm{\sum_{j=2}^N(B_\hbar)_jw_{1j}(B_\hbar)_j }\leq \Big\|B_\hbar |w_N(x_1-\cdot)|B_\hbar\Big\|_{\gS^1}
\leq \frac{C}{N}\sup_{x_1\in\R^d}\int_{B(x_1,R)}|w_N|.
\end{equation*}
When $\eta>0$ the supremum can be bounded by $\int_{\R^d}|w_N|=\int_{\R^d}|w|$, since we assume that $w\in L^1( \R^d)$ in this case. When $\eta=0$ (hence $w_N=w$) this can be controlled by
$$\sup_{x_1\in\R^d}\int_{B(x_1,R)}|w|\leq |B_R|\|w_2\|_{L^\ii(\R^d)}+|B_R|^{1+\frac2d}\norm{w_1}_{L^{1+\frac{d}2}(\R^d)}$$
since $w=w_1+w_2\in L^{1+d/2}+L^\ii(\R^d)$. In all cases, we have proved that 
$$\tr\bigg(\Gamma_{N,\hbar,\beta}\frac{1}{N}\sum_{1\leq j\neq   k\leq N}(B_\hbar)_jw_{jk}\big((B_\hbar)_j+(B_\hbar)_k\big)\bigg)=O(N^{-1}).$$
It then remains to look at 
\begin{multline*}
\left|(N-1)\tr\bigg(\Gamma_{N,\hbar,\beta}(B_\hbar)_{j'}(B_\hbar)_2w_{12}\bigg)\right|\nn\\
\leq  (N-1)\sqrt{\tr\bigg(\Gamma_{N,\hbar,\beta}(B_\hbar)_{j'}(B_\hbar)_2|w_{12}|(B_\hbar)_2(B_\hbar)_{j'}\bigg)}\sqrt{\tr(\Gamma_{N,\hbar,\beta}|w_{12}|)},
\end{multline*}
where $j'\in\{1,2\}$. The first term is estimated as before by
$$\tr\bigg(\Gamma_{N,\hbar,\beta}(B_\hbar)_{j'}(B_\hbar)_2|w_{12}|(B_\hbar)_2(B_\hbar)_{j'}\bigg)\leq \frac{C}{N^3}\sup_{x\in\R^d}\int_{B(x,R)}|w_N|.$$
The supremum is uniformly bounded. Hence
\begin{equation}
 \left|(N-1)\tr\bigg(\Gamma_{N,\hbar,\beta}(B_\hbar)_{j'}(B_\hbar)_2w_{12}\bigg)\right|\leq \frac{C}{N^{1/2}}\sqrt{\tr(\Gamma_{N,\hbar,\beta}|w_{12}|)}.
 \label{eq:estim_interaction_commutator}
\end{equation}
The estimate on $\tr(\Gamma_{N,\hbar,\beta}|w_{12}|)$ depends on the value of $\eta$. If $\eta=0$, then $w_N=w$ and we have $\tr(\Gamma_{N,\hbar,\beta}|w_{12}|)=O(1)$ by the Lieb-Thirring inequality. If $\eta>1/d$, we have assumed that $w\geq0$, hence $\tr(\Gamma_{N,\hbar,\beta}|w_{12}|)=\tr(\Gamma_{N,\hbar,\beta}w_{12})$ is uniformly bounded since this term appears in the energy. Finally, when $0<\eta<1/d$, the Lieb-Thirring inequality implies 
$$\tr(\Gamma_{N,\hbar,\beta}|w_{12}|)\leq C\norm{w_N}_{L^{1+d/2}(\R^d)}=CN^{d\eta \frac{d/2}{1+d/2}}\norm{w}_{L^{1+d/2}(\R^d)}.$$
When inserted in~\eqref{eq:estim_interaction_commutator}, we obtain an error of the order $N^{-\frac12+\frac{d\eta}{2}\frac{1}{1+d/2}}\to0$.
This concludes the proof of Lemma~\ref{lem:commutators}.
\end{proof}

At this point we have finished the proof of Theorems~\ref{theo:TL_Can_Sc} and~\ref{theo:dilute_limit} for the exact $N$-particle Gibbs states $\Gamma_{N,\hbar,\beta}$. It is possible to handle approximate Gibbs states using the relative entropy and Pinsker's inequality as we did for the one-particle Husimi functions. Indeed, consider a sequence of states $\Gamma_N$ such that 
$$\mathcal{E}_{\mathrm{Can}}^{N,\hbar} \myp{\Gamma_N}=e_{\mathrm{Can}}^{\beta} \myp{\hbar,N} +o(1).$$
We can write
$$\mathcal{E}_{\mathrm{Can}}^{N,\hbar} \myp{\Gamma_N}-e_{\mathrm{Can}}^{\beta} \myp{\hbar,N}= \frac{1}{\beta}\,\mathcal{H}(\Gamma_N,\Gamma_{N,\hbar,\beta})$$
where
$\mathcal{H}(A,B)=\tr(A(\log A-\log B)$
if the relative entropy. From the quantum Pinsker inequality $\mathcal{H}(A,B)\geq \norm{A-B}^2_{\mathfrak{S}^1}/2$ we infer that 
$$\tr\big|\Gamma_N-\Gamma_{N,\hbar,\beta}\big|\longrightarrow0$$
in trace norm. Since $\|m^{(k)}_{f,\Gamma}\|_{L^\ii(\R^{2dk})}\leq \tr|\Gamma|$ by~\cite[Eq.~(1.15)]{FouLewSol-18}, we conclude that 
$$\norm{m^{(k)}_{f,\Gamma_N}-m^{(k)}_{f,\Gamma_{N,\hbar,\beta}}}_{L^\ii(\R^{2dk})}\longrightarrow0.$$
Therefore $m^{(k)}_{f,\Gamma_N}$ has the same weak limit as the exact Gibbs state. The proof of Theorems~\ref{theo:TL_Can_Sc} and~\ref{theo:dilute_limit} is now complete.


\section{Proof of Theorem~\ref{theo:min_Vlasov}: study of the semiclassical functional} 
\label{sec:proof_theo_Vlasov_func}

This section is devoted to the proof of \cref{theo:min_Vlasov} and some auxiliary results on the semiclassical functional. We begin our analysis with the free particle case ($w=0$)  and then generalize to systems with pair interaction. We recall that the magnetic potential does not affect the energy, only the minimizer, and can be removed by a change of variables so we do not consider it here. For this section and for $\rho >0$ we denote by
\begin{equation*}
	S_{\mathrm{Vla}} \myp{\rho} = \bigg\{ m \in L^1 \big( \mathbb{R}^{2d} \big) \: \Big\rvert \: 0 \leq m \leq 1, \ \tfrac{1}{\myp{2 \pi}^d} \int_{\R^{2d}} m = \rho \bigg\}.
\end{equation*}
the set of admissible semi-classical measures.

\subsection{The free gas}


\begin{prop}[Minimizing the free semi-classical energy]
	\label{lem:functmin}
	Suppose that $ w = 0 $, and that $ V_+ \in L_{\loc}^1 \myp{\mathbb{R}^d} $ satisfies $ \int_{\mathbb{R}^d} e^{-\beta V_+\myp{x}} \, \mathrm{d} x < \infty $ for some $\beta >0$ and $V_- \in L^{d/2}(\mathbb{R}^d)\cap L^{1+d/2}(\mathbb{R}^d)$.
	Fix $ \rho > 0 $ and define $ m_0 \in \Ssc \myp{\rho} $ by
	\begin{equation*}
		m_0 \myp{x,p} := \frac{1}{1+ e^{\beta \myp{p^2 + V\myp{x}-\mu}}},
	\end{equation*}
	where $\mu$ is the unique chemical potential such that 
	$$\frac1{(2\pi)^{d}}\iint_{\R^{2d}}m_0(x,p)\,dx\,dp=\rho.$$
	Then
	\begin{align}
		e_{\mathrm{Vla}}^{\beta,w=0} \myp{\rho}
		&= \mathcal{E}_{\mathrm{Vla}}^{\beta,\rho,w=0} \myp{m_0} \nn \\
		&= - \frac{1}{\myp{2 \pi}^d \beta} \int_{\mathbb{R}^{2d}} \log \myp{1 + e^{- \beta \myp{p^2 + V\myp{x} - \mu}}} \, \mathrm{d} x \, \mathrm{d} p + \mu \rho.
	\label{eq:freefunctminimum}
	\end{align}
\end{prop}

\begin{proof}
The map 
$$R:=\mu\mapsto (2\pi)^{-d}\iint_{\R^{2d}}m_0(x,p)\,dx\,dp$$
is well-defined on $\R$, using that 
$$\frac{1}{1+ e^{\beta \myp{p^2 + V\myp{x}-\mu}}}\leq \frac{\max(1,e^{\beta\mu})}{1+e^{\beta \myp{p^2 + V\myp{x}}}}$$
which is integrable under our conditions on $V$, by the remarks after Theorem~\ref{theo:min_Vlasov}.
In addition, $R$ is increasing and continuous with 
$$\lim_{\mu\to-\ii}R(\mu)=0,\qquad \lim_{\mu\to+\ii}R(\mu)=+\ii.$$
Therefore we can always find $\mu$ so that the density of $m_0$ equals the given $\rho$. Note then that 
	\begin{equation*}
		1 - m_0 \myp{x,p} = e^{\beta\myp{p^2 + V\myp{x} -\mu}} m_0 \myp{x,p} = \frac{1}{1 + e^{-\beta\myp{p^2 + V\myp{x} -\mu}}},
	\end{equation*}
	so that
	\begin{align*}
		\mathcal{E}_{\mathrm{Vla}}^{\beta,\rho, w=0} \myp{m_0}
		&= \frac{1}{\myp{2 \pi}^d \beta} \int_{\mathbb{R}^{2d}}\bigg\{ \beta \myp{p^2 + V\myp{x} - \mu} m_0 + m_0 \log m_0\\
		&\qquad\qquad\qquad - m_0 \log \myp{e^{\beta \myp{p^2 + V\myp{x} -\mu}} m_0}\bigg\} \, \mathrm{d} x \, \mathrm{d} p \nonumber \\
		&\qquad+ \frac{1}{\myp{2 \pi}^{d} \beta} \int_{\mathbb{R}^{2d}}\left( \log\myp{1-m_0} + \beta \mu \, m_0\right) \, \mathrm{d} x \, \mathrm{d} p \nonumber \\
		&= - \frac{1}{\myp{2\pi}^d \beta} \int_{\mathbb{R}^{2d}} \log \myp{1 + e^{- \beta \myp{p^2 + V\myp{x} - \mu}}} \, \mathrm{d} x \, \mathrm{d} p + \mu \rho,
	\end{align*}
	showing the second equality in \eqref{eq:freefunctminimum}.
	That $ m_0 $ is the minimizer follows from the fact that the free energy is strictly convex. For instance, for any other $ m \in  S_{\mathrm{Vla}} \myp{\rho} $, since the function $ s \myp{t} = t \log t + \myp{1-t} \log \myp{1-t} $ is convex on $ \myp{0,1} $ with derivative $ s' \myp{t} = \log \big( \frac{t}{1-t} \big) $, we have pointwise
	\begin{align}
		s \myp{m} &\geq s \myp{m_0} + s' \myp{m_0} \myp{m-m_0} \nonumber \\
		&= -\beta \myp{p^2 + V\myp{x} -\mu} m + \beta \myp{p^2 + V\myp{x} -\mu} m_0 + s\myp{m_0},
	\label{eq:sbound}
	\end{align}
	replacing $m_0$ by its expression implies that $ \mathcal{E}_{\mathrm{Vla}}^{\beta,\rho,w=0} \myp{m} \geq \mathcal{E}_{\mathrm{Vla}}^{\beta,\rho,w=0} \myp{m_0} $. 
	That $ m_0 $ is the unique minimizer follows from the fact that $ \mathcal{E}_{\mathrm{Vla}}^{\beta,\rho,w=0} $ is a strictly convex functional.
\end{proof}

\begin{rem}
	\label{rem:semiclfunct}
	For an arbitrary domain $ \Omega \subseteq \mathbb{R}^{2d} $, we have by the very same arguments that 
	\begin{multline*}
\min_{\substack{m \in L^1 \myp{\Omega}\\ 0 \leq m \leq 1}}\bigg\{\frac{1}{\myp{2 \pi}^d} \int_{\Omega} \bigg( \myp{p^2+V\myp{x}} m \myp{x,p}  \, \mathrm{d} x + \frac{1}{\beta} s \myp{m \myp{x,p}} \bigg) \, \mathrm{d} x \, \mathrm{d} p\bigg\}\\
= - \frac{1}{\myp{2 \pi}^d \beta} \int_{\Omega} \log \myp{1+ e^{-\beta \myp{p^2 + V \myp{x}}} } \, \mathrm{d} x \, \mathrm{d} p.
	\end{multline*}
with the unique minimizer $ \widetilde{m_0} \myp{x,p} = (1+e^{\beta (p^2 + V\myp{x})})^{-1} $ and no chemical potential since we have dropped the mass constraint.
\end{rem}


\subsection{The interacting gas}


We now deal with the interacting case. When $w \neq 0$, to retrieve the existence of minimizers as well as their expression, we need to use compactness techniques and compute the Euler-Lagrange equation. We divide the proof in several lemmas. We start by proving the semi-continuity of the functional in \cref{lem:strongcont} and then prove the existence of minimizers on $\Ssc(\rho)$ in \cref{lem:functminexist}. To obtain the form of the minimizers we compute the Euler-Lagrange equation but because the entropy $s$ is not differentiable in $0$ and $1$ we first need to prove in \cref{lem:minimizerbounds} that minimizers cannot be equal to $0$ nor $1$ in sets of non zero measure. The proof of \cref{theo:min_Vlasov} is given at the end of this subsection.
\begin{lem}
\label{lem:strongcont}
	Fix $ \rho,\beta_0  > 0 $.
	Suppose that $ w =0 $, and that $ V_+ \in L_{\loc}^{1} \myp{\mathbb{R}^d}$, $V_{-} \in L^{d/2}(\mathbb{R}^d)\cap L^{1+d/2}(\mathbb{R}^d) $ satisfies $ \int_{\mathbb{R}^d} e^{-\beta_0 V_+\myp{x}} \id x < \infty $.
	Then for all $\beta > \beta_0$, $ \Esc^{\beta,\rho,w=0} $ is $L^1$-strongly lower semi-continuous on $ \Ssc \myp{\rho} $.
\end{lem}
\begin{proof}
	We have to show that for any $ C_0 \in \mathbb{R} $ 
	$$ \mathcal{L}(C_0) := \left\{ m \in \Ssc \myp{\rho} \mid \Esc^{\beta,\rho,w=0} \myp{m} \leq C_0\right\} $$
	 is closed with respect to the $ L^1 $-norm on $ \Ssc \myp{\rho} $.
	Let $ \myp{m_n} \subseteq \mathcal{L}(C_0) $ be a sequence converging towards some $ m \in L^1 \myp{\mathbb{R}^d} $ with respect to the $ L^1 $-norm.
	By the $ L^1 $ convergence we immediately have $ \frac{1}{\myp{2 \pi}^d} \iint_{\mathbb{R}^{2d}} m = \rho $, we can also extract a subsequence converging almost everywhere and obtain $ 0 \leq m \leq 1 $.
Applying \cref{rem:semiclfunct} with $ \Omega = \myt{\abs{x} + \abs{y} \geq R} $, we have for any $ R > 0 $ that
\begin{align}
	\frac{1}{\myp{2 \pi}^d} \iint_{\abs{x}+\abs{p} \geq R} \myp{p^2 + V \myp{x}} m_n \myp{x,p} + \frac{1}{\beta} s \myp{m_n \myp{x,p}} \id x \id p \nn \\ 
	\geq - \frac{1}{\beta} \iint_{\abs{x}+\abs{p} \geq R} \log \myp{1+e^{-\beta\myp{p^2 + V\myp{x}}}} \id x \id p = o_R(1).
\label{ineq:outside_a_ball}
\end{align}
	Now we use that $(m_n)$ is bounded in $L^\infty(\mathbb{R}^{2d})$ to obtain that $m_n \to m$ in $L^{p}(\mathbb{R}^{2d})$ for all $1\leq p <\infty$. By Fatou's lemma and dominated convergence we obtain
\begin{align*}
&\liminf_{n\to\infty} \iint_{\abs{x} + \abs{p} \leq R} \myp{p^2 + V_+\myp{x}} m_n \myp{x,p} \id x \id p \\
& \qquad \qquad \qquad \qquad \qquad \qquad \qquad  \geq  \iint_{\abs{x} + \abs{p} \leq R} \myp{p^2 + V_+\myp{x}} m\myp{x,p} \id x \id p, \\
&\iint_{\abs{x} + \abs{p} \leq R} V_-\myp{x} m_n \myp{x,p} \id x \id p \underset{n\to\infty}{\longrightarrow} \iint_{\abs{x} + \abs{p} \leq R} V_-\myp{x} m\myp{x,p} \id x \id p.
\end{align*}
It remains to deal with the entropy term: by continuity of $s$ and by dominated convergence we have 
\begin{equation*}
\iint_{\abs{x} + \abs{p} \leq R} s \myp{m_{n} \myp{x,p}} \id x \id p \underset{n\to\infty}{\longrightarrow} \iint_{\abs{x} + \abs{p} \leq R} s \myp{m \myp{x,p}}\id x \id p.
\end{equation*}
All in all we obtain
\begin{align*}
C_0 &\geq \liminf_{n\to\infty} \Esc^{\beta,\rho,w=0} \myp{m_n} \\ 
&\geq \frac{1}{\myp{2 \pi}^d} \iint_{\abs{x}+\abs{p} \leq R} \myp{p^2 + V\myp{x}} m\myp{x,p}\id x \id p \\
	&\qquad \qquad \qquad \qquad \qquad \qquad \quad + \frac{1}{\beta} \iint_{\abs{x} + \abs{p} \leq R} s \myp{m \myp{x,p}}\id x \id p + o(R) \\
&\geq \frac{1}{\myp{2 \pi}^d} \iint_{\abs{x}+\abs{p} \leq R} \myp{p^2 + V_+\myp{x}} m\myp{x,p}\id x \id p + o(R) \\
& \qquad \quad - \frac{1}{\myp{2 \pi}^d} \iint_{\mathbb{R}^{2d}} V_-\myp{x} m\myp{x,p}\id x \id p  + \frac{1}{\beta} \iint_{\mathbb{R}^{2d}} s \myp{m \myp{x,p}}\id x \id p.
\end{align*}
Finally, we use the monotone convergence theorem and let $R$ tend to $\infty$ to obtain $ \Esc^{\beta,\rho,w=0} \myp{m} \leq C_0$.
\end{proof}
\begin{lem}
\label{lem:functminexist}
	Fix $ \rho,\beta_0 > 0 $.
	Suppose that $ w \in L^{1+d/2} \myp{\mathbb{R}^d} + L_{\varepsilon}^{\infty} \myp{\mathbb{R}^d} + \mathbb{R}_+ \delta_0 $, $ V_+ \in L_{\loc}^1 \myp{\mathbb{R}^d}, V_{-} \in  L^{1+d/2}(\mathbb{R}^d)$ satisfies $ \int_{\mathbb{R}^d} e^{-\beta_0 V_+\myp{x}} \id x < \infty $ and $V_+(x)\to\infty$ as $|x|\to\infty$.
	Then for all $\beta > \beta_0$, $ \Esc^{\beta,\rho} $ is bounded below and has a minimizer $ m_0 $ in $ \Ssc \myp{\rho} $. 
\end{lem}

\begin{proof}
	Let $ \myp{m_n} \subseteq \Ssc \myp{\rho} $ be a minimizing sequence, i.e. $ \Esc^{\beta,\rho} \myp{m_n} \to \esc \myp{\rho} $ as $n\to\infty$.
	Since $ \myp{m_n} $ is bounded in both $ L^1 \myp{\mathbb{R}^{2d}} $ and $ L^{\infty} \myp{\mathbb{R}^{2d}} $, one can verify that up to extraction the sequence has a weak limit $ m_0 \in L^1 \myp{\mathbb{R}^{2d}} \cap L^{\infty} \myp{\mathbb{R}^{2d}} $ satisfying
	\begin{equation}
	\label{eq:weaklim}
		\int_{\mathbb{R}^{2d}} m_n \myp{x,p} \varphi \myp{x,p} \id x \id p 
		\to \int_{\mathbb{R}^{2d}} m_0 \myp{x,p} \varphi \myp{x,p} \id x \id p
	\end{equation}
	for any $ \varphi \in L^1 \myp{\mathbb{R}^{2d}} + L_{\varepsilon}^{\infty} \myp{\mathbb{R}^{2d}} $. Moreover, the weak limit $m_0$ satisfies $ 0\leq m_0 \leq 1 $ and $ \int_{\mathbb{R}^{2d}} m_0 \leq \rho \myp{2 \pi}^d $. Note that we do not have pointwise convergence a priori.
	Let us prove that $ m_0 $ is a minimizer of $ \Esc^{\beta,\rho} $ in $ \Ssc \myp{\rho} $. Our first step is to show the tightness of the sequence of probability measures $(m_n)$ to obtain $\int_{\mathbb{R}^{2d}} m_0 = (2\pi)^d \rho$, then we argue that $m_0 \in  \Ssc \myp{\rho}$ and minimizes $\Esc^{\beta,\rho}$ using weak lower-semicontinuity.

We start out by bounding the interaction term using some of the kinetic energy.
	Let $\varepsilon > 0$ and let us write $w=w_1 + w_2 + a \delta_0$ with $w_1 \in L^{1+d/2}(\mathbb{R}^d)$, $\|w_2\|_{L^\infty(\mathbb{R}^d)} < \varepsilon$ and $a\geq0$. We use Young's inequality to bound the interaction term
	\begin{align}
	\int_{\mathbb{R}^d} w\ast \rho_{m_n} \rho_{m_n} 
		&\geq \|w_1\|_{L^{1+d/2}(\mathbb{R}^d)} \|\rho_{m_n}\|_{L^{1+2/d}(\mathbb{R}^d)} \|\rho_{m_n}\|_{L^1(\mathbb{R}^d)} \nn \\
		& \qquad \qquad \qquad \qquad \qquad \qquad + \|w_2\|_{L^\infty(\mathbb{R}^d)} \|\rho_{m_n}\|_{L^1(\mathbb{R}^d)}^2  \nn \\
		&\geq C\varepsilon \iint_{\mathbb{R}^{2d}} p^2 m_n \myp{x,p} \id x \id p - C.
		\label{eq:wyoung}
	\end{align}
	In the last inequality we have used the well-known fact \cite{LieLos-01} that 
\begin{align}
	\label{eq:rhobound}
		\int_{\mathbb{R}^{d}} p^2 m \myp{x,p} \id p&\geq \inf_{\substack{ 0\leq \widetilde{m} \leq 1 \\ \int \widetilde{m} = \myp{2 \pi}^d \rho_m \myp{x}}} \int_{\mathbb{R}^d} p^2 \widetilde{m} \myp{p} \id p \nn\\
		&= \myp{2 \pi}^d c_{\mathrm{TF}} \frac{d}{d+2} \rho_m \myp{x}^{1+ 2/d},
\end{align}
which gives the Lieb-Thirring inequality for classical measures on phase space.
Similarly we have
\begin{equation}\label{eq:v_minus}
\int_{\mathbb{R}^d} V_-(x) \rho_{m_n}(x) dx  \leq C\left( \varepsilon^{-d/2} \|V_-\|^{1+d/2}_{L^{1+d/2}(\mathbb{R}^d)} + \varepsilon\|\rho_{m_n}\|^{1+2/d}_{L^{1+2/d}(\mathbb{R}^d)}\right).
\end{equation}
	Now using \cref{lem:functmin}, \eqref{eq:rhobound}, \eqref{eq:wyoung}  and \eqref{eq:v_minus}, denoting $\alpha = (\beta- \beta_0)/(2\beta)$  we have
	\begin{align}
	\label{eq:functbound}
		C \geq \Esc^{\beta,\rho} \myp{m_n}
		&\geq \frac{\alpha}{\myp{2 \pi}^d} \iint_{\mathbb{R}^{2d}} \myp{p^2 + V \myp{x}} m_n + \frac{1}{2 \rho} \int_{\mathbb{R}^{d}} \myp{w \ast \rho_{m_n}} \rho_{m_n} \nn \\
		&\quad\quad + \frac{1}{2} e_{\rm Vla}^{\beta(1-\alpha),w=0} \myp{\rho} \nn \\
		&\geq \frac{\alpha - C \varepsilon}{\myp{2 \pi}^d} \iint_{\mathbb{R}^{2d}} \myp{p^2 + V_+ \myp{x}} m_n	 - C 
	\end{align}
	Note that by construction, $\beta(1-\alpha) > \beta_0$. Taking $\varepsilon >0$ sufficiently small but positive, the above inequality shows the tightness condition 
	\begin{eqnarray}
	\label{eq:tightness}
	\iint_{\mathbb{R}^{2d}} \myp{p^2 + V_+ \myp{x}} m_n \myp{x,p} \id x \id p \leq C. 
	\end{eqnarray}
	Therefore $\iint_{\mathbb{R}^{2d}}m_0 = (2\pi)^d \rho$. 
	
	Now we prove that $\liminf_{n\to\infty} \Esc^{\beta,\rho}(m_n) \geq \Esc^{\beta,\rho}(m_0)$. From the tightness condition it is easy to verify that $\rho_{m_n} \rightharpoonup \rho_{m_0}$ and that 
	\begin{equation*}
	\int_{\mathbb{R}^d} (w-a\delta_0)\ast \rho_{m_n} \rho_{m_n} \to \int_{\mathbb{R}^d} (w-a\delta_0)\ast \rho_{m_0} \rho_{m_0}.
	\end{equation*}
	To finish, we deal with the delta part of the interaction as well as the entropy part. We use that a continuous convex function is always weakly lower semi-continuous. We obain
	\begin{align*}
	a &\int_{\mathbb{R}^d}\rho_{m_0}^2 = \int_{\mathbb{R}^d} \lim_{n\to\infty} \rho_{m_n}^2 \leq \lim_{n\to\infty} \int_{\mathbb{R}^d} \rho_{m_n}^2, \\
	&\int_{\mathbb{R}^d} s(m_0) = \int_{\mathbb{R}^d} \lim_{n\to\infty} s(m_{n}) \leq \liminf_{n\to\infty}\int_{\mathbb{R}^d} s(m_{n}).
	\end{align*}
\end{proof}
\begin{lem}
	\label{lem:minimizerbounds}
	Fix $ \rho,\beta_0 > 0 $.
	Suppose that $ w \in L^{1+d/2} \myp{\mathbb{R}^d} + L_{\varepsilon}^{\infty} \myp{\mathbb{R}^d} + \mathbb{R}_+ \delta_0$, $ V_+ \in L_{\loc}^1 \myp{\mathbb{R}^d}, V_{-} \in L^{1+d/2}(\mathbb{R}^d)$ satisfies $ \int_{\mathbb{R}^d} e^{-\beta_0 V_+\myp{x}} \id x < \infty $ and $V_+(x)\to\infty$ as $|x|\to\infty$.
	Then any minimizer $ m_0 \in \Ssc \myp{\rho} $ of $ \Esc^{\beta,\rho} $ satisfies 
	\begin{equation*}
	0 < m(x,p) < 1 \quad \textrm{for } (x,p)\in\mathbb{R}^{2d} \: \textrm{almost everywhere.}
	\end{equation*}
\end{lem}
\begin{proof} Define $\Omega_{1} := \myt{m_0 =1} $ and $\Omega_{0} := \myt{m_0 = 0}$. Our goal is to prove that $ \Omega_{1}$ and $ \Omega_{0}$ have $0$ measure. To this end, we will first show that $|\Omega_{1}| |\Omega_{0}| = 0$. Then we use that at least one of then is a null set to prove that so is the other one. Let us first assume neither of them are null sets. Let $r>0$, $ 0 < \lambda < \frac{1}{2} $ and for almost every $(\xi_1,\xi_2)\in\Omega_{1} \times \Omega_{0}$ define
	\begin{equation*}
	\varphi_{1} = \lambda \mathds{1}_{B \myp{\xi_1,r}\cap \Omega_{1}}, \quad
	\varphi_{2} = \lambda \mathds{1}_{B\myp{\xi_2,r'}\cap \Omega_{0}},
	\end{equation*}
	where $r' := \min \myt{s \geq 0 \mid \abs{B\myp{\xi_2,s}\cap \Omega_{0}} = \abs{B\myp{\xi_1,r}\cap \Omega_{1}} }$. We will use the notation $v(r) = \abs{B\myp{\xi_1,r}\cap \Omega_{1}}$. Note that by Lebesgue's density theorem, for almost every $(\xi_1,\xi_2)\in\Omega_{1} \times \Omega_{0}$ we have $v(r)>0$ and $r'<\infty$. The idea is to consider the function $ m_0-\varphi_{1}+ \varphi_{2} \in \Ssc \myp{\rho} $ and use the fact that $ m_0 $ is a minimizer of $ \Esc^{\beta,\rho} $ to obtain a contradiction.
	Let us estimate the entropy, using that $s(0) = s(1) = 0$ and $ s (t) = s (1-t) $, we obtain
	\begin{align*}
		\iint_{\mathbb{R}^{2d}} s \myp{m_0 - \varphi_1 + \varphi_2}
		&= \iint_{\mathbb{R}^{2d}} s \myp{m_0} + s( \varphi_1) +  s(\varphi_2) \\
		&=  2 s \myp{\lambda} v(r)+ \iint_{\mathbb{R}^{2d}} s \myp{m_0}.
	\end{align*}
	It remains to estimate the contribution of this small perturbation to interaction energy, we have
	\begin{multline*}
	\int_{\mathbb{R}^d} \rho_{m_0-\varphi_1 + \varphi_2} w \ast \rho_{m_0-\varphi_1 + \varphi_2} 
	= \int_{\mathbb{R}^d} \rho_{m_0} w \ast \rho_{m_0} + 2  \int_{\mathbb{R}^d} \rho_{\varphi_2 - \varphi_1} w \ast \rho_{m_0} \\ +  \int_{\mathbb{R}^d}\rho_{\varphi_2 - \varphi_1} w\ast  \rho_{\varphi_2 - \varphi_1}.
	\end{multline*}
Let $\varepsilon >0$ and let us write $w=w_1 + w_2 + a \delta_0$ with $w_1 \in L^{1+d/2}(\mathbb{R}^d)$, $\|w_2\|_{L^\infty(\mathbb{R}^d)} < \varepsilon$ and $a\geq 0$. We first use Young's inequality to bound the last term
	\begin{align*}
		\int_{\mathbb{R}^d} &w \ast \myp{\rho_{\varphi_2} - \rho_{\varphi_1}} \myp{\rho_{\varphi_2} - \rho_{\varphi_1}} \\
		&\leq \normL{w_1}{1 + d/2} \normL{\rho_{\varphi_2} - \rho_{\varphi_1}}{1} \normL{\rho_{\varphi_2} - \rho_{\varphi_1}}{1 + 2/d} \\
		&\qquad \qquad \qquad \qquad + \|w_2\|_{L^{\infty}_\varepsilon(\mathbb{R}^d)} \|\rho_{\varphi_2} - \rho_{\varphi_1}\|_{L^{1}(\mathbb{R}^d)}^2 + a \|\rho_{\varphi_2} - \rho_{\varphi_1}\|_{L^2(\mathbb{R}^d)}^2 \\
		&\leq C\lambda^2 \bigg(\normL{w}{1+d/2}v(r)^{1+\frac{d}{d+2}} +  \|w_2\|_{L^{\infty}_\varepsilon(\mathbb{R}^d)}v(r)^2 + a v(r) \bigg).
	\end{align*}
	Next and similarly we estimate the second term (minus the delta interaction)
	\begin{align*}
		\int_{\mathbb{R}^d} (w_1&+w_2) \ast \rho_{m_0} \myp{\rho_{\varphi_2} - \rho_{\varphi_1}} \\
		&\leq \normL{w_1}{1+d/2}\normL{\rho_{m_0}}{1+2/d} \normL{\rho_{\varphi_2} - \rho_{\varphi_1}}{1}  \\
		& \qquad \qquad \qquad \qquad \qquad + \|w_2\|_{L^\infty_\varepsilon(\mathbb{R}^d)} \|\rho_{m_0}\|_{L^1(\mathbb{R}^d)} \|\rho_{\varphi_2} - \rho_{\varphi_1}\|_{L^1(\mathbb{R}^d)} \\
		&\leq C \lambda ( \normL{w_1}{1+d/2}\normL{\rho_{m_0}}{1 +2/d} + \|w_2\|_{L^\infty_\varepsilon(\mathbb{R}^d)}\|\rho_{m_0}\|_{L^1(\mathbb{R}^d)}  ) v(r).
	\end{align*}
	Since $ m_0 $ is a minimizer, these estimates imply that
	\begin{align*}
		\Esc^{\beta,\rho} \myp{m_0} 
		&\leq \Esc^{\beta,\rho} \myp{m_0 - \varphi_1 + \varphi_2} \\
		&\leq \Esc^{\beta,\rho} \myp{m_0} + \frac{1}{\myp{2 \pi}^d} \iint_{\mathbb{R}^{2d}} \myp{p^2 + V \myp{x} + a\rho_{m_0}} \myp{\varphi_2 - \varphi_1}  \\
		&\quad  + C\lambda^2 \left(\normL{w}{1+d/2}v(r)^{1+\frac{d}{d+2}} +  \|w_2\|_{L^{\infty}_\varepsilon(\mathbb{R}^d)}v(r)^2 + a v(r) \right) \\
		&  \quad + C \lambda \bigg( \normL{w_1}{1+d/2}\normL{\rho_{m_0}}{1 +2/d} \\ 
		& \qquad \qquad \qquad \qquad + \|w_2\|_{L^\infty_\varepsilon(\mathbb{R}^d)}\|\rho_{m_0}\|_{L^1(\mathbb{R}^d)}\bigg)v(r)  + \frac{2  s \myp{\lambda}}{\myp{2 \pi}^d\beta} v(r).
	\end{align*}
	Now we divide the last inequality by $v(r) $ and we let $ r $ tend to zero and use the Lebesgue differentiation theorem (and the Lebesgue density theorem), to obtain that for almost all $(\xi_1,\xi_2)\in\Omega_{1} \times \Omega_{0}$
	\begin{multline*}
		- \frac{2s(\lambda)}{\lambda \beta} \leq - p_1^2 - V \myp{x_1} - a\rho_{m_0}(x_1) + p_2^2 + V \myp{x_2} + a	\rho_{m_0}(x_2) \\ + C  \normL{w}{1+d/2}\normL{\rho_{m_0}}{1+2/d} .
	\end{multline*}
	Now letting $\lambda$ tend to zero, we have that for almost all $(\xi_1,\xi_2)\in\Omega_{1} \times \Omega_{0}$, $p_2^2 + V \myp{x_2}  + a\rho_{m_0}(x_2)  - p_1^2 - V \myp{x_1} - a\rho_{m_0}(x_1) = \infty$ which, since $V \in L^{1+d/2}_{\loc}(\mathbb{R}^d)$ and $\rho_{m_0} \in L^{1+2/d}_{\loc}(\mathbb{R}^d)$, implies that $|\Omega_{1} \times \Omega_{0}| = 0$. Therefore, at least one of them is a null set, we will treat the case where $|\Omega_{0}| = 0$ and $|\Omega_{1}| \neq 0$, the other one can be dealt with similarly. Because $m$ has finite mass we can find $\varepsilon >0$ such that $\Omega_{2,\varepsilon} := \left\{ 1-\varepsilon \leq m(x,p) \leq 1-\varepsilon/2\right\}$ is not a null set. Defining $\varphi_1$ and $\varphi_2$ (replacing $\Omega_{0}$ by $\Omega_{2,\varepsilon}$) as before and doing the same computations we obtain that for almost all $(\xi_1,\xi_2)\in\Omega_{1} \times \Omega_{2,\varepsilon}$
	\begin{multline*}
	- \frac{s(\lambda)}{\lambda \beta} \leq- p_1^2 - V \myp{x_1} - a\rho_{m_0}(x_1) + p_2^2 + V \myp{x_2} + a\rho_{m_0}(x_2)  \\ + \frac{s(m(\xi_2) - \lambda) - s(m(\xi_2))}{\lambda}+ C  \normL{w}{1+d/2}\normL{\rho_{m_0}}{1+2/d}.
	\end{multline*}
	Because $s$ is continuously differentiable on $[1-2\varepsilon,1-\varepsilon/2 ]$, the difference quotient above is bounded uniformly in $\xi_2 \in \Omega_{2,\varepsilon}$ and $\lambda>0$ small enough. Letting $\lambda$ tend to zero, we end up with the same contradiction as before showing that $\Omega_{1}$ is a null set.
\end{proof}

\begin{proof}[Proof of \cref{theo:min_Vlasov}]
	We assume $A=0$ without loss of generality, since it can be removed by a change of variable. 
	
	We will first show that the expression \eqref{eq:m0def} of the minimizers is correct by computing the Euler-Lagrange equation associated with any such minimizer $ m_0 $. This gives automatically the expression of the minimum energy (\ref{eq:functminimum}). We conclude, in the case $\widehat{w}\geq 0$, by showing that the chemical potential $\mu$ is given by (\ref{eq:chemical_potential_derivative}).
	
Let $\varepsilon >0$ small enough and $ \varphi \in L^{1}\cap L^\infty(\{\varepsilon < m < 1- \varepsilon\})$ such that $ \iint \varphi= \myp{2 \pi}^d \rho $. For $ \delta = \frac{\varepsilon}{1+ \left\lVert \varphi \right\rVert_{\infty}} $ we have $ m_t := \frac{m_0+t\varphi}{1+t} \in S_{\mathrm{Vla}} \myp{\rho} $ for all $ t \in \myp{-\delta,\delta} $.
	Since $ m_0 $ is a minimizer, we must have $ \frac{\mathrm{d}}{\mathrm{d} t} \mathcal{E}_{\mathrm{Vla}}^{\beta} \myp{m_t}_{\lvert t=0} = 0 $. Using that $ \frac{\mathrm{d}}{\mathrm{d} t} m_t =  (\varphi - m_0)\myp{1+t}^{-2} $ and $ s' \myp{t} = \log \big( \frac{t}{1-t} \big) $ we obtain
	\begin{align}
		\iint_{\mathbb{R}^{2d}} &\myp{p^2 + V \myp{x} + \frac{1}{\rho} w \ast \rho_{m_0} \myp{x} + \frac{1}{\beta} \log \myp{\frac{m_0 \myp{x,p}}{1-m_0 \myp{x,p}}} } \varphi \myp{x,p}  \, \mathrm{d} x \, \mathrm{d} p \nn \\
		&= \iint_{\mathbb{R}^{2d}} \bigg(p^2 + V \myp{x} + \frac{1}{\rho} w \ast \rho_{m_0} \myp{x} \nn \\ &\qquad\qquad\qquad\qquad \label{eq:chemical_potential_def}
		 + \frac{1}{\beta} \log \myp{\frac{m_0 \myp{x,p}}{1-m_0 \myp{x,p}}} \bigg) m_0 \myp{x,p} \, \mathrm{d} x \, \mathrm{d} p.
	\end{align}
	Denoting the right hand side by $ \myp{2 \pi}^d \mu_{\mathrm{Vla}} \myp{\rho} \rho $, we have shown for any $ \varphi$ verifying the above conditions that
	\begin{multline*}
		\iint_{\{\varepsilon < m < 1-\varepsilon\}} \bigg( p^2 + V \myp{x} + \frac{1}{\rho} w \ast \rho_{m_0} \myp{x} \\ + \frac{1}{\beta} \log \myp{\frac{m_0 \myp{x,p}}{1-m_0 \myp{x,p}}} - \mu_{\mathrm{Vla}} \myp{\rho} \bigg) \varphi \myp{x,p}  \, \mathrm{d} x \, \mathrm{d} p
		= 0.
	\end{multline*}
	This is enough for the left factor in the integrand above to be zero almost everywhere on $\{\varepsilon < m < 1-\varepsilon\}$. But $\varepsilon$ can be taken arbitrary small and by \cref{lem:minimizerbounds} we have $\bigcup_{\varepsilon>0}\{\varepsilon < m < 1-\varepsilon\} {=  \{0<m<1\}} = \mathbb{R}^{2d}$ almost everywhere, from which we obtain \eqref{eq:m0def}.
	
	That $\rho_{m_0} \in L^{2}(\mathbb{R}^d) \cap L^{1+d/2}(\mathbb{R}^d)$ follows from \cref{lem:rho_regularity} and the fact that $m_0$ satisfies \eqref{eq:m0def}.
	
	It remains to prove \eqref{eq:chemical_potential_derivative} when it is assumed that $\widehat{w}\geq 0$. This is a classical argument and we only sketch it, we refer to \cite{LieSim-77b} for further details. First note that the assumption $\widehat{w}\geq 0$ ensures the convexity of $\Esc^{\beta,\rho}$, hence for $\rho' >0$, $F_{\mathrm{Vla}}^{\beta}(\rho',\rho)$ is the minimum of a convex function under a linear constraint, it is therefore convex. This implies that, for $\rho' >0$, the function $F_{\mathrm{Vla}}^{\beta}(\cdot,\rho')$ is continuous on $\mathbb{R}_+$ and continuously differentiable except maybe in a countable number of values of $\rho$. We first show that 
	\begin{equation*} 
	\mathbb{R}^*_+ \ni \rho \mapsto \mu(\rho) \in \mathbb{R}
	\end{equation*} 
	defines a bijection, where $\mu(\rho)$, defined in (\ref{eq:m0def}), is the Lagrange multiplier associated to the constraint $\rho$. Consider, for $\mu \in \mathbb{R}$, the unconstrained minimization problem 
	\begin{equation}
	\label{eq:unconstrained}
	 \inf_{0\leq m\leq 1} \Esc^{\beta,\rho'}(m) - \frac{\mu}{(2\pi)^d} \iint_{\mathbb{R}^{2d}} m  = \inf_{\rho\geq 0} F_{\mathrm{Vla}}^{\beta}(\rho,\rho') - \mu \rho.
	\end{equation} 
 This yields a minimizer $m^\mu$ and hence a density $\rho(\mu) := (2\pi)^{-d}\iint m^\mu$, see \cref{rem:semiclfunct}. The expression of $m^\mu$ can be computed through the Euler-Lagrange equation, 
 \begin{equation*}
 m^\mu = \frac{1}{1+e^{\beta(p^2 + V + \rho'^{-1} \rho_{m^\mu} \ast w - \mu)}}
 \end{equation*}
From (\ref{eq:unconstrained}), the density $m^\mu$ must also satisfy $\Esc^{\beta,\rho'}(m^\mu) =  F_{\mathrm{Vla}}^{\beta}(\rho(\mu),\rho')$ and since $\widehat{w} \geq 0$, we conclude that $m^\mu$ is also the unique solution of this equation and must satisfy (\ref{eq:m0def}) where $\mu(\rho)$ appears. By identification, $\mu = \mu(\rho)$ is the Lagrange multiplier associated to the minimization problem at density $\rho$. This proves the bijective correspondance between $\mu(\rho)$ and $\rho$.

Finally, if $F_{\mathrm{Vla}}^{\beta}(\cdot,\rho')$ is differentiable in some $\rho_0$, the above discussion shows (\ref{eq:chemical_potential_derivative}) for $\rho = \rho_0$. But because of the one-to-one correspondance between $\mu$ and $\rho$, $\partial_{\rho}F_{\mathrm{Vla}}^{\beta}$ cannot be discontinuous, this concludes the proof.
\end{proof}


\begin{thebibliography}{10}

\bibitem{BacBrePetPicTza-15}
{\sc V.~{Bach}, S.~{Breteaux}, S.~{Petrat}, P.~{Pickl}, and T.~{Tzaneteas}},
  {\em Kinetic energy estimates for the accuracy of the time-dependent
  {H}artree-{F}ock approximation with {C}oulomb interaction}, J. Math. Pures
  Appl., in press (2015).

\bibitem{BacLieSol-94}
{\sc V.~Bach, E.~H. Lieb, and J.~P. Solovej}, {\em Generalized {H}artree-{F}ock
  theory and the {H}ubbard model}, J. Statist. Phys., 76 (1994), pp.~3--89.

\bibitem{BarGolGotMau-03}
{\sc C.~Bardos, F.~Golse, A.~D. Gottlieb, and N.~J. Mauser}, {\em Mean field
  dynamics of fermions and the time-dependent {H}artree-{F}ock equation}, J.
  Math. Pures Appl. (9), 82 (2003), pp.~665--683.

\bibitem{BarBuc-81}
{\sc M.~Barranco and J.-R. Buchler}, {\em Equation of state of hot, dense
  stellar matter: Finite temperature nuclear thomas-fermi approach}, Phys. Rev.
  C, 24 (1981), pp.~1191--1202.

\bibitem{BenJakPorSafSch-16}
{\sc N.~{Benedikter}, V.~{Jaksic}, M.~{Porta}, C.~{Saffirio}, and
  B.~{Schlein}}, {\em Mean-field evolution of fermionic mixed states}, Comm.
  Pure Appl. Math., 69 (2016), pp.~2250--2303.

\bibitem{BenNamPorSchSei-18}
{\sc N.~{Benedikter}, P.~Nam, M.~{Porta}, B.~{Schlein}, and R.~{Seiringer}},
  {\em Optimal upper bound for the correlation energy of a fermi gas in the
  mean-field regime}, ArXiv e-prints,  (2018), p.~arXiv:1809.01902.

\bibitem{BenPorSafSch-16}
{\sc N.~Benedikter, M.~Porta, C.~Saffirio, and B.~Schlein}, {\em From the
  {H}artree dynamics to the {V}lasov equation}, Arch. Ration. Mech. Anal., 221
  (2016), pp.~273--334.

\bibitem{BenPorSch-14}
{\sc N.~Benedikter, M.~Porta, and B.~Schlein}, {\em Mean-field evolution of
  fermionic systems}, Comm. Math. Phys., 331 (2014), pp.~1087--1131.

\bibitem{BroKos-90}
{\sc L.~G. Brown and H.~Kosaki}, {\em Jensen's inequality in semi-finite von
  {N}eumann algebras}, J. Operator Theory, 23 (1990), pp.~3--19.

\bibitem{BruBucJorLom-68}
{\sc K.~A. Brueckner, J.~R. Buchler, S.~Jorna, and R.~J. Lombard}, {\em
  Statistical theory of nuclei}, Phys. Rev., 171 (1968), pp.~1188--1195.

\bibitem{CowAsh-57}
{\sc R.~D. Cowan and J.~Ashkin}, {\em Extension of the thomas-fermi-dirac
  statistical theory of the atom to finite temperatures}, Phys. Rev., 105
  (1957), pp.~144--157.

\bibitem{CycFroKirSim-87}
{\sc H.~L. Cycon, R.~G. Froese, W.~Kirsch, and B.~Simon}, {\em Schr{\"o}dinger
  operators with application to quantum mechanics and global geometry}, Texts
  and Monographs in Physics, Springer-Verlag, Berlin, study~ed., 1987.

\bibitem{DieRadSch-18}
{\sc E.~Dietler, S.~Rademacher, and B.~Schlein}, {\em From {H}artree dynamics
  to the relativistic {V}lasov equation}, J. Stat. Phys., 172 (2018),
  pp.~398--433.

\bibitem{ElgErdSchYau-04}
{\sc A.~Elgart, L.~Erd{\H o}s, B.~Schlein, and H.-T. Yau}, {\em Nonlinear
  {H}artree equation as the mean field limit of weakly coupled fermions}, J.
  Math. Pures Appl., 83 (2004), pp.~1241--1273.

\bibitem{FeyMetTel-49}
{\sc R.~P. Feynman, N.~Metropolis, and E.~Teller}, {\em Equations of state of
  elements based on the generalized fermi-thomas theory}, Phys. Rev., 75
  (1949), pp.~1561--1573.

\bibitem{FouLewSol-18}
{\sc S.~{Fournais}, M.~{Lewin}, and J.~P. {Solovej}}, {\em The semi-classical
  limit of large fermionic systems}, Calc. Var. Partial Differ. Equ.,  (2018),
  pp.~57--105.

\bibitem{FroKno-11}
{\sc J.~Fr{\"o}hlich and A.~Knowles}, {\em A microscopic derivation of the
  time-dependent {H}artree-{F}ock equation with {C}oulomb two-body
  interaction}, J. Stat. Phys., 145 (2011), pp.~23--50.

\bibitem{GilPee-55}
{\sc J.~J. Gilvarry and G.~H. Peebles}, {\em Solutions of the
  temperature-perturbed thomas-fermi equation}, Phys. Rev., 99 (1955),
  pp.~550--552.

\bibitem{GoiPitStr-08}
{\sc S.~Giorgini, L.~P. Pitaevskii, and S.~Stringari}, {\em Theory of ultracold
  atomic fermi gases}, Rev. Mod. Phys., 80 (2008), pp.~1215--1274.

\bibitem{Gottlieb-05}
{\sc A.~D. Gottlieb}, {\em Examples of bosonic de {F}inetti states over finite
  dimensional {H}ilbert spaces}, J. Stat. Phys., 121 (2005), pp.~497--509.

\bibitem{GraMajSchTex-18}
{\sc A.~Grabsch, S.~N. Majumdar, G.~Schehr, and C.~Texier}, {\em Fluctuations
  of observables for free fermions in a harmonic trap at finite temperature},
  SciPost Phys., 4 (2018), p.~14.

\bibitem{HaiPorRex-18}
{\sc C.~{Hainzl}, M.~{Porta}, and F.~{Rexze}}, {\em On the correlation energy
  of the mean-field fermi gas}, ArXiv e-prints,  (2018), p.~arXiv:1806.11411.

\bibitem{Latter-55}
{\sc R.~Latter}, {\em Temperature behavior of the thomas-fermi statistical
  model for atoms}, Phys. Rev., 99 (1955), pp.~1854--1870.

\bibitem{LewNamRou-16}
{\sc M.~Lewin, P.~Nam, and N.~Rougerie}, {\em Bose gases at positive
  temperature and non-linear {G}ibbs measures}, in Proceedings of the
  International Congress of Mathematical Physics, 2015.
\newblock ArXiv e-prints.

\bibitem{LewNamRou-18a}
{\sc M.~Lewin, P.~Nam, and N.~Rougerie}, {\em Gibbs measures based on {1D}
  (an)harmonic oscillators as mean-field limits}, J. Math. Phys., 59 (2018),
  p.~041901.

\bibitem{LewNamRou-14}
{\sc M.~Lewin, P.~T. Nam, and N.~Rougerie}, {\em Derivation of {H}artree's
  theory for generic mean-field {B}ose systems}, Adv. Math., 254 (2014),
  pp.~570--621.

\bibitem{LewNamRou-15}
\leavevmode\vrule height 2pt depth -1.6pt width 23pt, {\em Derivation of
  nonlinear {G}ibbs measures from many-body quantum mechanics}, J. {\'E}c.
  polytech. Math., 2 (2015), pp.~65--115.

\bibitem{LewNamRou-18c}
\leavevmode\vrule height 2pt depth -1.6pt width 23pt, {\em Classical field
  theory limit of {2D} many-body quantum {G}ibbs states}, ArXiv e-prints,
  (2018).

\bibitem{LewNamRou-18b}
\leavevmode\vrule height 2pt depth -1.6pt width 23pt, {\em The interacting {2D}
  bose gas and nonlinear gibbs measures}, in Gibbs measures for nonlinear
  dispersive equations, B.~S. {Giuseppe Genovese} and V.~Sohinger, eds., 2018.
\newblock Oberwolfach mini-workshop.

\bibitem{LewNamSerSol-15}
{\sc M.~Lewin, P.~T. Nam, S.~Serfaty, and J.~P. Solovej}, {\em Bogoliubov
  spectrum of interacting {B}ose gases}, Comm. Pure Appl. Math., 68 (2015),
  pp.~413--471.

\bibitem{LieLos-01}
{\sc E.~H. Lieb and M.~Loss}, {\em Analysis}, vol.~14 of Graduate Studies in
  Mathematics, American Mathematical Society, Providence, RI, 2nd~ed., 2001.

\bibitem{LieSeiSol-05}
{\sc E.~H. Lieb, R.~Seiringer, and J.~P. Solovej}, {\em Ground-state energy of
  the low-density {F}ermi gas}, Phys. Rev. A, 71 (2005), p.~053605.

\bibitem{LieSim-77}
{\sc E.~H. Lieb and B.~Simon}, {\em The {H}artree-{F}ock theory for {C}oulomb
  systems}, Commun. Math. Phys., 53 (1977), pp.~185--194.

\bibitem{LieSim-77b}
\leavevmode\vrule height 2pt depth -1.6pt width 23pt, {\em The {T}homas-{F}ermi
  theory of atoms, molecules and solids}, Adv. Math., 23 (1977), pp.~22--116.

\bibitem{LieThi-75}
{\sc E.~H. Lieb and W.~E. Thirring}, {\em Bound on kinetic energy of fermions
  which proves stability of matter}, Phys. Rev. Lett., 35 (1975), pp.~687--689.

\bibitem{LieThi-76}
\leavevmode\vrule height 2pt depth -1.6pt width 23pt, {\em Inequalities for the
  moments of the eigenvalues of the {S}chr{\"o}dinger hamiltonian and their
  relation to {S}obolev inequalities}, Studies in Mathematical Physics,
  Princeton University Press, 1976, pp.~269--303.

\bibitem{LieThi-84}
\leavevmode\vrule height 2pt depth -1.6pt width 23pt, {\em Gravitational
  collapse in quantum mechanics with relativistic kinetic energy}, Ann.
  Physics, 155 (1984), pp.~494--512.

\bibitem{LieYau-87}
{\sc E.~H. Lieb and H.-T. Yau}, {\em The {C}handrasekhar theory of stellar
  collapse as the limit of quantum mechanics}, Commun. Math. Phys., 112 (1987),
  pp.~147--174.

\bibitem{MadThese}
{\sc P.~S. {M}adsen}, {\em In preparation}, PhD thesis, Aarhus University,
  2019.

\bibitem{March-55b}
{\sc N.~H. March}, {\em Equations of state of elements from the thomas-fermi
  theory ii: Case of incomplete degeneracy}, Proc. Phys. Soc., 68 (1955),
  p.~1145.

\bibitem{MarBet-40}
{\sc R.~E. {Marshak} and H.~A. {Bethe}}, {\em The generalized thomas-fermi
  method as applied to stars.}, Astrophys. J., 91 (1940), p.~239.

\bibitem{NarSew-81}
{\sc H.~Narnhofer and G.~Sewell}, {\em Vlasov hydrodynamics of a quantum
  mechanical model}, Comm. Math. Phys., 79 (1981), pp.~9--24.

\bibitem{NarThi-81}
{\sc H.~Narnhofer and W.~Thirring}, {\em Asymptotic exactness of finite
  temperature {T}homas-{F}ermi theory}, Ann. Phys., 134 (1981), pp.~128 -- 140.

\bibitem{PetPic-16}
{\sc S.~Petrat and P.~Pickl}, {\em A new method and a new scaling for deriving
  fermionic mean-field dynamics}, Math. Phys. Anal. Geom., 19 (2016), pp.~Art.
  3, 51.

\bibitem{ReeSim1}
{\sc M.~Reed and B.~Simon}, {\em Methods of {M}odern {M}athematical {P}hysics.
  {I}. Functional analysis}, Academic Press, 1972.

\bibitem{Robinson-71}
{\sc D.~W. Robinson}, {\em The thermodynamic pressure in quantum statistical
  mechanics}, Springer-Verlag, Berlin-New York, 1971.
\newblock Lecture Notes in Physics, Vol. 9.

\bibitem{Rougerie-15}
{\sc N.~{Rougerie}}, {\em De finetti theorems, mean-field limits and
  {B}ose-{E}instein condensation}, ArXiv e-prints,  (2015).

\bibitem{Ruelle}
{\sc D.~Ruelle}, {\em Statistical mechanics. Rigorous results}, {Singapore:
  World Scientific. London: Imperial College Press }, 1999.

\bibitem{Schonhammer-17}
{\sc K.~Sch\"onhammer}, {\em Deviations from {W}ick's theorem in the canonical
  ensemble}, Phys. Rev. A, 96 (2017), p.~012102.

\bibitem{Seiringer-06b}
{\sc R.~Seiringer}, {\em The thermodynamic pressure of a dilute {F}ermi gas},
  Comm. Math. Phys., 261 (2006), pp.~729--757.

\bibitem{Simon-80}
{\sc B.~Simon}, {\em The classical limit of quantum partition functions}, Comm.
  Math. Phys., 71 (1980), pp.~247--276.

\bibitem{Spohn-81}
{\sc H.~Spohn}, {\em On the {V}lasov hierarchy}, Math. Methods Appl. Sci., 3
  (1981), pp.~445--455.

\bibitem{Thirring}
{\sc W.~E. Thirring}, {\em Quantum Mathematical Physics}, vol.~Atoms, Molecules
  and Large Systems, Springer, Second Edition 2002.

\bibitem{TriThese}
{\sc A.~{T}riay}, {\em In preparation}, PhD thesis, University of
  Paris-Dauphine, 2019.

\end{thebibliography}

\end{document}